\pdfoutput=1  %arXiv addition for pdflatex, permits no file extensions.
\documentclass[11pt]{article}
\usepackage{url}
\usepackage{epsf}
\usepackage{graphicx}
\usepackage{amsmath}
\usepackage{fullpage}
\newtheorem{lemma}{Lemma}
\newtheorem{theorem}{Theorem}

\newcommand{\qed}{\rule{0.5em}{1.5ex}}
\newcommand{\fqed}{{\hfill~\qed}}
\newenvironment{proof}{{\noindent \bf Proof.}}
                      {{\hfill \fqed} \vspace{1em}}
\newtheorem{proposition}{Proposition}

\newcommand{\Pa}{{\mathcal P}}

{\makeatletter
 \gdef\xxxmark{%
   \expandafter\ifx\csname @mpargs\endcsname\relax % in minipage?
     \expandafter\ifx\csname @captype\endcsname\relax % in figure/caption?
       \marginpar{xxx}% not in a caption or minipage, can use marginpar
     \else
       xxx % notice trailing space
     \fi
   \else
     xxx % notice trailing space
   \fi}
 \gdef\xxx{\@ifnextchar[\xxx@lab\xxx@nolab}
 \long\gdef\xxx@lab[#1]#2{{\bf [\xxxmark #2 ---{\sc #1}]}}
 \long\gdef\xxx@nolab#1{{\bf [\xxxmark #1]}}
 \gdef\turnoffxxx{\long\gdef\xxx@lab[##1]##2{}\long\gdef\xxx@nolab##1{}}%
}

\title{$\pi/2$-Angle Yao Graphs are Spanners}
\author{
Prosenjit Bose
    \thanks{School of Computer Science, Carleton University, Ottawa, Canada. \protect\url{jit@scs.carleton.ca}.
    Supported by NSERC.}
\and
Mirela Damian
    \thanks{Department of Computer Science, Villanova University, Villanova, USA.
    \protect\url{mirela.damian@villanova.edu}. Supported by NSF grant CCF-0728909.}
\and
Karim Dou\"ieb
    \thanks{School of Computer Science, Carleton University, Ottawa, Canada. \protect\url{kdouieb@ulb.ac.be}.
    Supported by NSERC.}
\and
Joseph O'Rourke
    \thanks{Department of Computer Science, Smith College, Northampton, USA. \protect\url{orourke@cs.smith.edu.}
    Supported by NSERC.}
\and
Ben Seamone
    \thanks{School of Mathematics and Statistics, Carleton University, Ottawa, Canada. \protect\url{bseamone@connect.carleton.ca}.}
\and
Michiel Smid
    \thanks{School of Computer Science, Carleton University, Ottawa, Canada. \protect\url{michiel@scs.carleton.ca.}
    Supported by NSERC.}
\and
Stefanie Wuhrer
    \thanks{NRC Institute for Information Technology, Ottawa, Canada. \protect\url{Stefanie.Wuhrer@nrc-cnrc.gc.ca}.}
}
\date{}

\begin{document}

\maketitle

\begin{abstract}
We show that the Yao graph $Y_4$ in the $L_2$ metric is a spanner with
stretch factor $8(29+23\sqrt{2})$. Enroute to this, we also show that
the Yao graph $Y^\infty_4$ in the $L_\infty$ metric is a planar spanner
with stretch factor $8$.
\end{abstract}

\section{Introduction}
Let $V$ be a finite set of points in the plane and let $G=(V,E)$ be
the complete Euclidean graph on $V$. We will refer to the points in
$V$ as \emph{nodes}, to distinguish them from other points in the
plane.
The \emph{Yao graph}~\cite{Yao82} with an integer parameter $k > 0$,
denoted $Y_k$, is defined as follows. At each node $u \in V$, any $k$
equally-separated rays originating at $u$ define $k$ cones.
In each cone, pick a shortest edge $uv$, if there is one, and add to $Y_k$
the directed edge $\overrightarrow{uv}$. Ties are broken arbitrarily.
Most of the time we ignore the direction of an edge $uv$; we refer
to the directed version $\overrightarrow{uv}$ of $uv$ only when its
origin ($u$) is important and unclear from the context.
We will distinguish between $Y_k$, the Yao graph in the Euclidean $L_2$
metric, and $Y^\infty_k$, the Yao graph in the $L_\infty$ metric. Unlike
$Y_k$ however, in constructing $Y^\infty_k$ ties are broken by always
selecting the most counterclockwise edge; the reason for this choice
will become clear in Section~\ref{sec:y4inf}.

For a given subgraph $H \subseteq G$ and a fixed $t \ge 1$, $H$ is
called a $t$-\emph{spanner} for $G$ if, for any two nodes
$u,v \in V$, the shortest path in $H$ from $u$ to $v$ is no longer than
$t$ times the length of $uv$.
The value $t$ is called the \emph{dilation} or the
\emph{stretch factor} of $H$. If $t$ is constant, then $H$ is called a
\emph{length spanner}, or simply a \emph{spanner}.

The class of graphs $Y_k$ has been much studied. Bose et al.~\cite{bmnsz-agbsp-03}
showed that, for $k \geq 9$, $Y_k$ is a spanner with stretch factor
$\frac{1}{\cos \frac{2\pi}{k} - \sin \frac{2\pi}{k}}$. %In Lemma~\ref{lem:Y6},
In the appendix, we improve the stretch factor and show that, in fact, 
$Y_k$ is a spanner for any $k \geq 7$.
%[JOR:
%\emph{The previous sentence appears to have dropped in from another paper...?}
%]
%Keil~\cite{Keil88}
%showed that $Y_k$ is a spanner with stretch factor
%$\frac{1}{1 - 2 \sin{\pi/k}}$, for $k \ge 7$. About a decade later,
%Fischer et al.~\cite{FLZ98}
%showed that $Y_4$ is a \emph{weak spanner} with stretch factor
%$t' = \sqrt{3 + \sqrt{5}}$, meaning that between any two nodes
%$u, v \in V$, there is a path in $Y_4$ that lies entirely
%inside a disk of radius $t' \cdot |uv|$ centered at $v$.
Recently, Molla~\cite{M09} showed that $Y_2$ and $Y_3$ are 
not spanners, and that $Y_4$ is a spanner with stretch factor 
$4(2+\sqrt{2})$, for the special case when the nodes in $V$ are 
in convex position (see also ~\cite{DMP09}). 
%It has also been established that $Y_3$ is not a spanner,
The authors conjectured that $Y_4$ is a spanner for arbitrary point 
sets.
In this paper, we settle their conjecture and prove that $Y_4$ is a
spanner with stretch factor $8(29+23\sqrt{2})$.

%% MS: added a brief overview
%%We start with a few definitions and notation, then present some
%%preliminary lemmas in section, and finally we use these lemmas in
%%deriving our main result in Section~\ref{sec:Yinf+Y4}.
The paper is organized as follows.
In Section~\ref{sec:y4inf}, we prove that the graph $Y^\infty_4$ is a
spanner with stretch factor $8$.
In Section~\ref{secY4L2}, we prove, in a sequence of Lemmas, several
properties for the graph $Y_4$. Finally, in Section~\ref{sec:Yinf+Y4},
we use the properties of Section~\ref{secY4L2} to prove that for every
edge $ab$ in $Y^\infty_4$, there exists a path between $a$ and $b$ in
$Y_4$, whose length is not much more than the Euclidean distance between
$a$ and $b$. By combining this with the result of Section~\ref{sec:y4inf},
it follows that $Y_4$ is a spanner.

\section{$Y^\infty_4$: in the $L_\infty$ Metric}
\label{sec:y4inf}
In this section we focus on $Y^\infty_4$, which has a nicer structure
compared to $Y_4$. First we prove that $Y^\infty_4$ is planar. Then we
use this property to show that $Y^\infty_4$ is an $8$-spanner.
%% MS added following 3 lines to emphasize the metric we are talking
%% about.
To be more precise, we prove that for any two nodes $a$ and $b$,
the graph $Y^\infty_4$ contains a path between $a$ and $b$ whose
length (in the $L_{\infty}$-metric) is at most $8 |ab|_{\infty}$.

We need a few definitions.
We say that two edges $ab$ and $cd$ \emph{properly cross} (or \emph{cross},
for short) if they share a point other than an endpoint ($a$, $b$, $c$ or $d$);
we say that $ab$ and $cd$ \emph{intersect} if they share a point (either an
interior point or an endpoint).
%
%%%%%%%%%%%%%%%%%%%%%%%%%%%%%%%%%Figure Begin
\begin{figure}[htbp]
\centering
\includegraphics[width=0.6\linewidth]{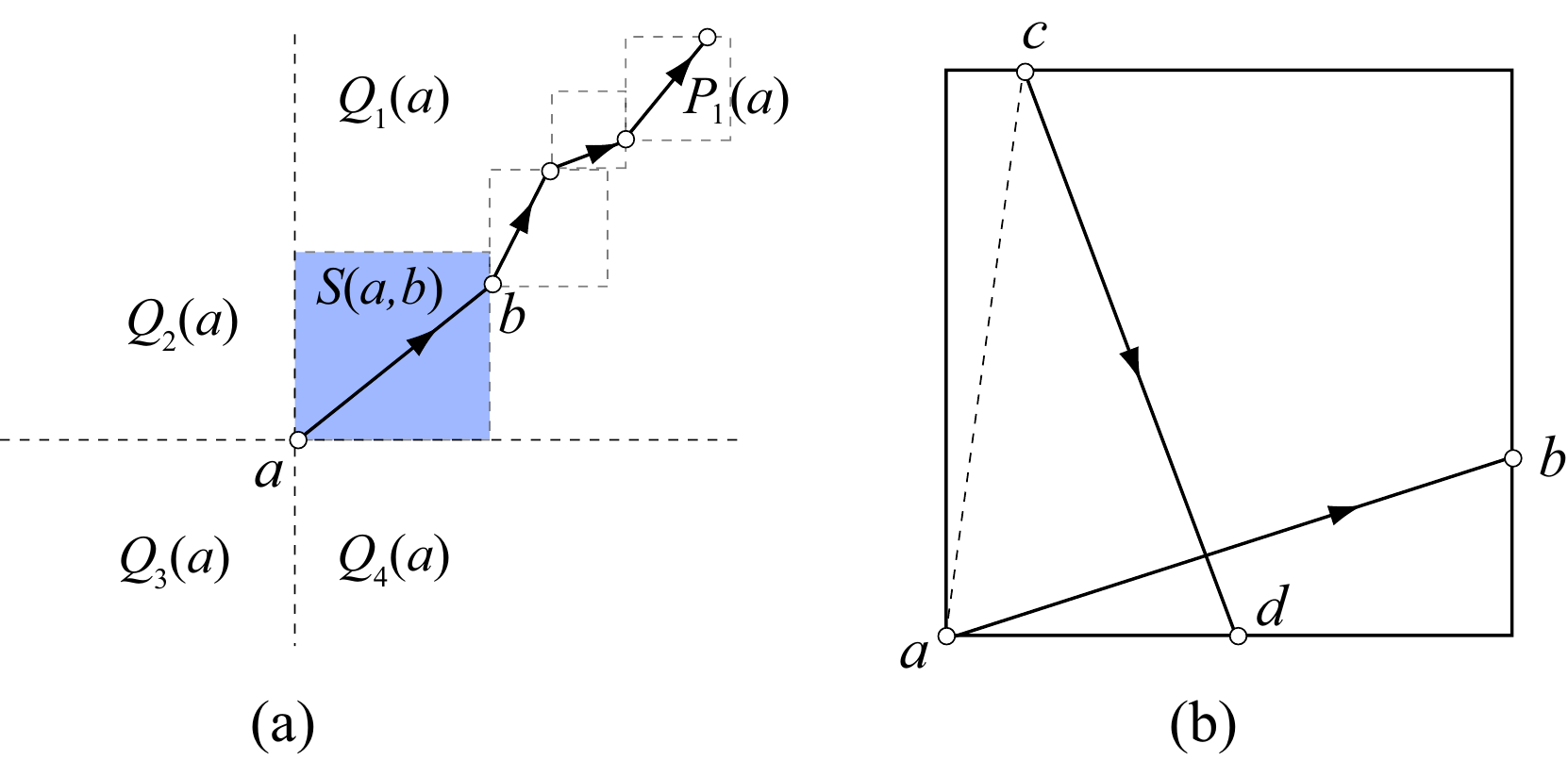}
\caption{(a) Definitions: $Q_i(a)$, $P_i(a)$ and $S(a,b)$. (b) Lemma~\ref{lem:y4infplanar}: $ab$ and $cd$ cannot cross.}
\label{fig:y4infplanar}
\end{figure}
%%%%%%%%%%%%%%%%%%%%%%%%%%%%%%%%%Figure End
%
Let $Q_1(a)$, $Q_2(a)$, $Q_3(a)$ and $Q_4(a)$ be
the four quadrants at $a$, as in Figure~\ref{fig:y4infplanar}a.
Let $P_i(a)$ be the path that starts at point $a$ and follows the directed
Yao edges in quadrant $Q_i$.
Let $P_i(a,b)$ be the subpath of $P_i(a)$ that starts at $a$ and ends at $b$.
Let $|ab|_\infty$ be the $L_\infty$ distance between $a$ and $b$.
Let $sp(a,b)$ denote a shortest path in $Y_4^\infty$ between $a$ and $b$.
Let $S(a,b)$ denote the open square with corner $a$ whose boundary
contains $b$, and let $\partial S(a,b)$ denote the boundary of $S(a,b)$.
%Let $S[a,b] = S(a,b) \cup \partial S(a,b)$.
These definitions are illustrated in Figure~\ref{fig:y4infplanar}a.
For a node $a \in V$, let $x(a)$ denote the $x$-coordinate of $a$ and
$y(a)$ denote the $y$-coordinate of $a$.

\begin{lemma}
$Y^\infty_4$ is planar.
\label{lem:y4infplanar}
\end{lemma}
\begin{proof}
The proof is by contradiction. Assume the opposite. Then there are
two edges $\overrightarrow{ab}, \overrightarrow{cd} \in Y^\infty_4$
that cross each other. Since $\overrightarrow{ab} \in Y^\infty_4$,
$S(a,b)$ must be empty of nodes in $V$, and similarly for $S(c,d)$. Let
$j$ be the intersection point between $ab$ and $cd$.
Then $j \in S(a,b) \cap S(c,d)$, meaning that $S(a,b)$ and $S(c,d)$ must
overlap. However, neither square may contain $a, b, c$ or $d$. It
follows that $S(a,b)$ and $S(c,d)$ coincide, meaning that $c$ and $d$
lie on $\partial S(a,b)$ (see Figure~\ref{fig:y4infplanar}b). Since $cd$
intersects $ab$, $c$ and $d$ must lie on opposite sides of $ab$.  Thus either $ac$ or
$ad$ lies counterclockwise from $ab$. Assume without loss of generality that
$ac$ lies counterclockwise from $ab$; the other case is identical.
Because $S(a,c)$ coincides with $S(a,b)$, we have that
$|ac|_\infty = |ab|_\infty$. In this case however, $Y^\infty_4$ would break
the tie between $ac$ and $ab$ by selecting the most counterclockwise edge,
%[JOR: \emph{Check me here.}]  %% MS Joe is correct
which is $\overrightarrow{ac}$. This contradicts the fact that
$\overrightarrow{ab} \in Y^\infty_4$.
\end{proof}

\noindent
It can be easily shown that each face of $Y^\infty_4$ is either a triangle
or a quadrilateral (except for the outer face). We skip this proof however,
since we do not make use of this property in this paper.

\begin{theorem}
$Y_4^\infty$ is an $8$-spanner.
\label{thm:y4infspanner}
\end{theorem}
\begin{proof}
We show that, for any pair of points $a, b \in V$, $|sp(a,b)|_\infty < 8|ab|_\infty$.
The proof is by induction on the pairwise distance between the points in $V$.
Assume without loss of generality that $b \in Q_1(a)$, and $|ab|_\infty = |x(b)-x(a)|$.
Consider the case in which $ab$ is a closest pair of points in $V$ (the base case
for our induction). If $ab \in Y_4^\infty$, then $|sp(a,b)|_\infty = |ab|_\infty$.
Otherwise, there must be $ac \in Y_4^\infty$, with
$|ac|_\infty = |ab|_\infty$. But then $|bc|_\infty < |ab|_\infty$
(see Figure~\ref{fig:y4inf}a), a contradiction.

%%%%%%%%%%%%%%%%%%%%%%%%%%%%%%%%%Figure Begin
\begin{figure}[htp]
\centering
\includegraphics[width=\linewidth]{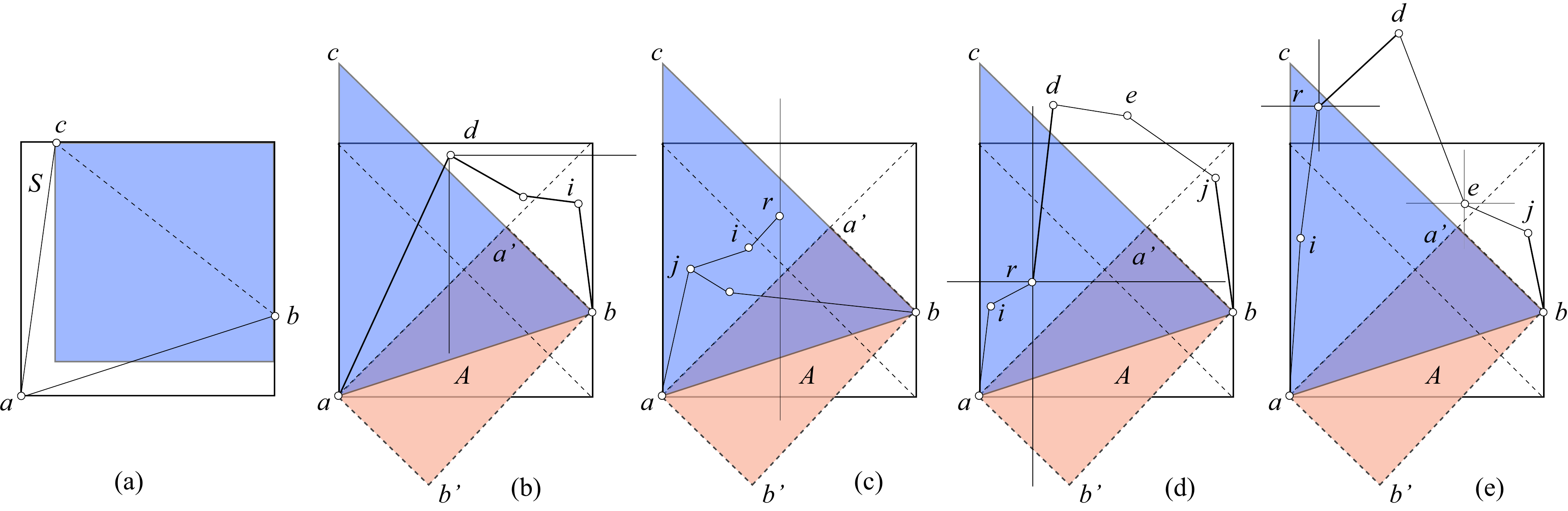}
\caption{(a) Base case. (b) $\triangle abc$ empty (c) $\triangle abc$ non-empty, $P_{ar} \cap P_2(b) = \{j\}$
(d) $\triangle abc$ non-empty, $P_{ar} \cap P_2(b) = \emptyset$, $e$ above $r$
(e) $\triangle abc$ non-empty, $P_{ar} \cap P_2(b) = \emptyset$, $e$ below $r$.}
\label{fig:y4inf}
\end{figure}
%%%%%%%%%%%%%%%%%%%%%%%%%%%%%%%%%Figure End

Assume now that the inductive hypothesis holds for all pairs of points
closer than $|ab|_\infty$. If $ab \in Y_4^\infty$, then $|sp(a,b)|_\infty = |ab|_\infty$
and the proof is finished. If $ab \notin Y_4^\infty$, then the square $S(a,b)$ %defined by $|ab|_\infty$
must be nonempty.

Let $A$ be the rectangle $ab'ba'$ as in Figure~\ref{fig:y4inf}b, where $ba'$ and
$bb'$ are parallel to the diagonals of $S$.
If $A$ is nonempty, then we can use induction to prove that $|sp(a,b)|_\infty <= 8|ab|_\infty$
as follows. Pick $c \in A$ arbitrary. Then $|ac|_\infty + |cb|_\infty = |x(c)-x(a)| + |x(b)-x(c)|
= |ab|_\infty$,
and by the inductive hypothesis $sp(a,c) \oplus sp(c,b)$ is a path in $Y_4^\infty$ no
longer than $8|ac|_\infty + 8|cb|_\infty = 8|ab|_\infty$; here $\oplus$ represents the 
concatenation operator.
%[JOR: \emph{First use of $\oplus$ operator here, now used before
% it is defined in the next section.}]
Assume now that $A$ is empty.
Let $c$ be at the intersection between the line supporting $ba'$ and the
vertical line through $a$ (see Figure~\ref{fig:y4inf}b).
We discuss two cases, depending on whether
$\triangle abc$ is empty of points or not.

\paragraph{Case 1: } $\triangle abc$ is empty of points. Let $ad \in P_1(a)$. We show that
$P_4(d)$ cannot contain an edge crossing $ab$. Assume the opposite, and let
$st \in P_4(d)$ cross $ab$. Since $\triangle abc$ is empty, $s$ must lie above
$bc$ and $t$ below $ab$,  therefore
$|st|_\infty \ge |y(s)-y(t)| > |y(s)-y(b)| = |sb|_\infty$,
contradicting the fact that $st \in Y_4^\infty$. It follows that
$P_4(d)$ and $P_2(b)$ must meet in a point $i \in P_4(d) \cap P_2(b)$ (see Figure~\ref{fig:y4inf}b).
Now note that
%% MS switched b and i
%% $|P_4(d, i)\oplus P_2(i,b)|_\infty
$|P_4(d, i)\oplus P_2(b,i)|_\infty
    \le |x(d)-x(b)| + |y(d)-y(b)| < 2|ab|_\infty$.
Thus we have that
\[
  |sp(a, b)|_\infty \le |ad \oplus P_4(d, i)\oplus P_2(b,i)|_\infty
  < |ab|_\infty + 2|ab|_\infty = 3|ab|_\infty.
\]
\paragraph{Case 2: } $\triangle abc$ is nonempty.
In this case, we seek a short path from $a$ to $b$ that does not
cross to the underside of $ab$. This is to avoid oscillating paths
that cross $ab$ arbitrarily many times. Let $r$ be the rightmost point
that lies inside $\triangle abc$. Arguments similar to the ones used in Case~1 show that
$P_3(r)$ cannot cross $ab$ and therefore it must meet $P_1(a)$ in a point $i$. Then
$P_{ar} = P_1(a, i) \oplus P_3(r, i)$ is a path in $Y_4^\infty$ of length
\begin{equation}
|P_{ar}|_\infty < |x(a)-x(r)| + |y(a)-y(r)| < |ab|_\infty + 2|ab|_\infty = 3|ab|_\infty.
%< 2|ab|_\infty
\label{eq:par}
\end{equation}
The term $2|ab|_\infty$ in the inequality above represents the fact that
$|y(a) - y(r)| \le |y(a) - y(c)| \le 2|ab|_\infty$.
Consider first the simpler situation in which $P_2(b)$ meets $P_{ar}$ in a
point $j \in P_2(b) \cap P_{ar}$ (see Figure~\ref{fig:y4inf}c).
Let $P_{ar}(a, j)$ be the subpath of $P_{ar}$ extending between $a$ and $j$.
Then $P_{ar}(a, j) \oplus P_2(b,j)$ is a path in $Y_4^\infty$ from $a$ to $b$,
therefore
\[
|sp(a,b)|_\infty \le |P_{ar}(a, j) \oplus P_2(b, j)|_\infty <
2 |y(j)-y(a)| + |ab|_\infty \le 5|ab|_\infty.
%4|ab|_\infty.
\]

Consider now the case when $P_2(b)$ does not intersect $P_{ar}$. We argue that, in this
case, $Q_1(r)$ may not be empty. Assume the opposite. Then
no edge $st \in P_2(b)$ may cross $Q_1(r)$. This is because,
for any such edge, $|sr|_\infty < |st|_\infty$, contradicting
$st \in Y_4^\infty$. This implies that $P_2(b)$ intersects $P_{ar}$, again
a contradiction to our assumption.

We have established that $Q_1(r)$ is nonempty. Let $rd \in P_1(r)$.
The fact that $P_2(b)$ does not intersect $P_{ar}$ implies that $d$
lies to the left of $b$.
The fact that $r$ is the rightmost point in $\triangle abc$ implies that
$d$ lies outside $\triangle abc$ (see Figure~\ref{fig:y4inf}d). It also implies that
$P_4(d)$ shares no points with $\triangle abc$. This along with
arguments similar to the ones used in case 1 show that
$P_4(d)$ and $P_2(b)$ meet in a point $j \in P_4(d) \cap P_2(b)$.
Thus we have found a path
\begin{equation}
P_{ab} = P_1(a,i) \oplus P_3(r, i) \oplus rd \oplus P_4(d, j) \oplus P_2(b, j)
\label{eq:pabinf}
\end{equation}
extending from $a$ to $b$ in $Y_4^\infty$. If $|rd|_\infty = |x(d) - x(r)|$, then
$|rd|_\infty < |x(b)-x(a)| = |ab|_\infty$, and the path $P_{ab}$ has length
\begin{equation}
|P_{ab}|_\infty \le 2 |y(d)-y(a)| + |ab|_\infty < 7|ab|_\infty.
\label{eq:Pab}
\end{equation}
In the above, we used the fact that
$|y(d)-y(a)| = |y(d)-y(r)| + |y(r)-y(a)| < |ab|_\infty +2|ab|_\infty$.
Suppose now that
\begin{equation}
|rd|_\infty = |y(d) - y(r)|.
\label{eq:rdlen}
\end{equation}
In this case, it is unclear whether the path $P_{ab}$ defined by~(\ref{eq:pabinf}) is short,
since $rd$ can be arbitrarily long compared to $ab$. Let $e$ be the clockwise
neighbor of $d$ along the path $P_{ab}$ ($e$ and $b$ may coincide).
Then $e$ lies below $d$, and either $de \in P_4(d)$, or $ed \in P_2(e)$ (or both).
\begin{enumerate}
\item
If $e$ lies above $r$, or at the same level as $r$ (i.e., $e \in Q_1(r)$, as in Figure~\ref{fig:y4inf}d), then
\begin{eqnarray}
|y(e) - y(r)| & < & |y(d) - y(r)|
\label{eq:yrd}
\end{eqnarray}
Since $rd \in P_1(r)$ and $e$ is in the same quadrant of $r$ as $d$, we have
$|rd|_\infty \le |re|_\infty$. This along with inequalities~(\ref{eq:rdlen})
and~(\ref{eq:yrd}) implies $|re|_\infty > |y(e) - y(r)|$, which in turn
implies $|re|_\infty = |x(e) - x(r)| \le |ab|_\infty$, and so
$|rd|_\infty \le |ab|_\infty$.
Then inequality~(\ref{eq:Pab}) applies here
as well, showing that $|P_{ab}|_\infty < 7|ab|_\infty$.

\item
If $e$ lies below $r$ (as in Figure~\ref{fig:y4inf}e), then
\begin{equation}
|ed|_\infty \ge |y(d)-y(e)| \ge |y(d)-y(r)| = |rd|_\infty.
\label{eq:rde}
\end{equation}
Assume first that $ed \in P_2(e)$, or $|ed|_\infty = |x(e)-x(d)|$. In either case,
\[
  |ed|_\infty \le |er|_\infty < 2|ab|_\infty.
\]
This along with inequality~(\ref{eq:rde}) shows that $|rd|_\infty < 2|ab|_\infty$.
Substituting this upper bound in~(\ref{eq:pabinf}), we get
\[
|P_{ab}|_\infty \le 2 |y(d)-y(a)| + 2|ab|_\infty < 8|ab|_\infty.
\]
Assume now that $ed \not\in P_2(e)$, and $|ed|_\infty = |y(e)-y(d)|$. Then
$ee' \in P_2(e)$ cannot go above $d$ (otherwise $|ed|_\infty < |ee'|_\infty$,
contradicting $ee' \in P_2(e)$). This along with the fact $de \in P_4(d)$ implies
that $P_2(e)$ intersects $P_{ar}$ in a point $k$. Redefine
\[
P_{ab} = P_{ar}(a,k) \oplus P_2(e, k) \oplus P_4(e, j) \oplus P_2(b,j)
\]
Then $P_{ab}$ is a path in $Y_4^\infty$ from $a$ to $b$ of length
\[
|P_{ab}| \le 2 |y(r)-y(a)| + |ab|_\infty \le 5|ab|_\infty.
\]

\end{enumerate}
We have established that $|sp(a,b)|_\infty \le |P_{ab}|_\infty <  8|ab|_\infty$.
This concludes the proof.
\end{proof}

\noindent
This theorem will be employed in Section~\ref{sec:Yinf+Y4}.

\section{$Y_4$: in the $L_2$ Metric}
\label{secY4L2}
In this section we establish basic properties of $Y_4$. The ultimate goal of this section is
to show that, if two edges in $Y_4$ cross, there is a short path between their endpoints (Lemma~\ref{lem:recross}). We
begin with a few definitions.

Let $Q(a,b)$ denote the infinite quadrant with origin at $a$ that contains $b$.
For a pair of nodes $a, b \in V$, define recursively a directed path
$\Pa(a \rightarrow b)$ from $a$ to $b$ in $Y_4$ as follows.
If $a=b$, then $\Pa(a \rightarrow b) = null$.
If $a \neq b$, there must exist $\overrightarrow{ac} \in Y_4$ that
lies in $Q(a,b)$. In this case, define
\[ \Pa(a \rightarrow b) =
   \overrightarrow{ac} \oplus \Pa(c \rightarrow b).
\]
Recall that $\oplus$ represents the concatenation operator.
This definition is illustrated in Figure~\ref{fig:defs}a.
Fischer et al.~\cite{FLZ98} show that $\Pa(a \rightarrow b)$ is well defined and
lies entirely inside the square centered at $b$ whose boundary contains $a$.
%$S(v, u)$.
%An immediate implication of this is that $Y_4$ is a \emph{weak} spanner.%~\cite{FLZ98}.

%%%%%%%%%%%%%%%%%%%%%%%%%%%%%%%%%Figure Begin
\begin{figure}[htbp]
\centering
\includegraphics[width=0.6\linewidth]{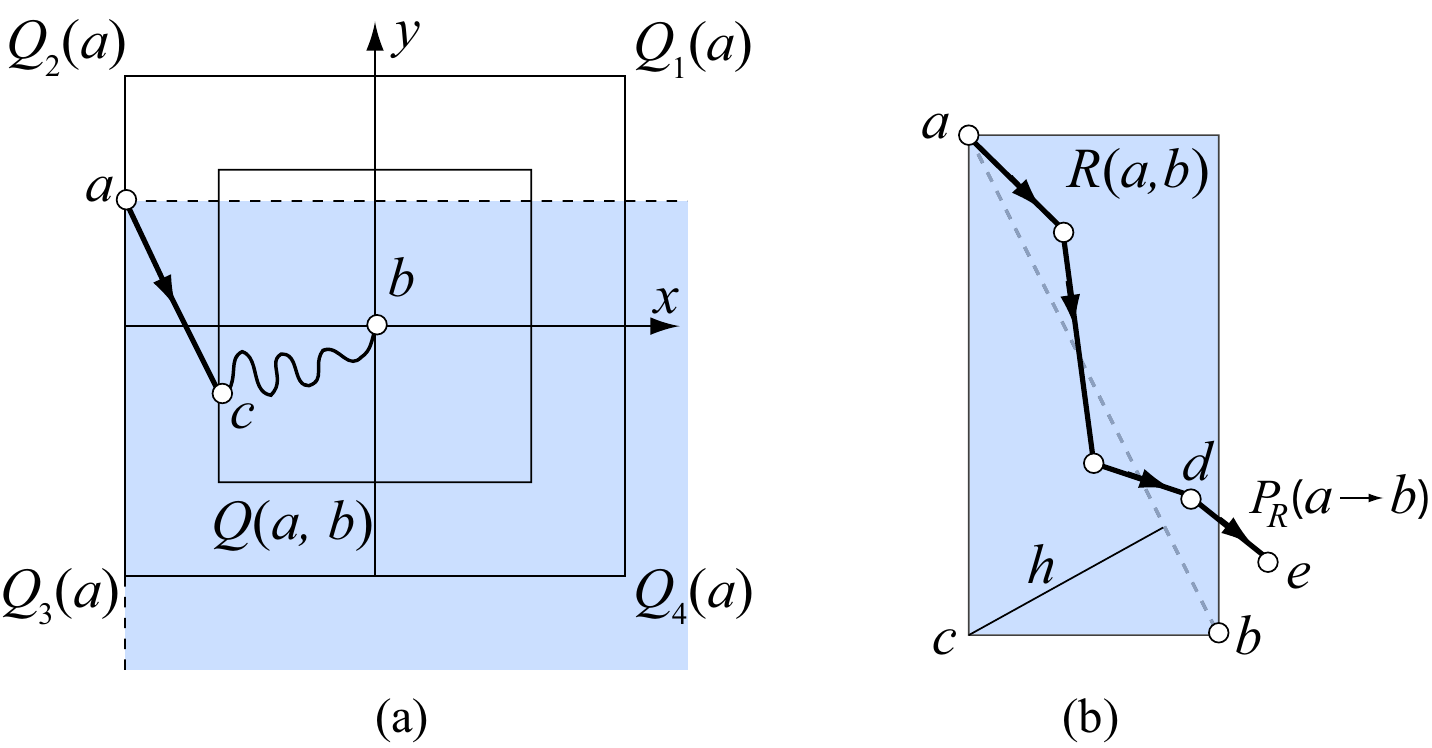}
\caption{Definitions. (a) $Q(a,b)$ and $\Pa(a \rightarrow b)$. (b) $\Pa_R(a \rightarrow b)$.}
\label{fig:defs}
\end{figure}
%%%%%%%%%%%%%%%%%%%%%%%%%%%%%%%%%Figure End

For any node $a \in V$, let $D(a,r)$ denote the open disk centered at
$a$ of radius $r$, and let $\partial D(a,r)$ denote the boundary of
$D(a, r)$. Let $D[a,r] = D(a,r) \cup \partial D(a,r)$.
For any path $P$ and any pair of nodes $a$ and $b$ on $P$, let $P[a,b]$
denote the subpath of $P$ that starts at $a$ and ends at $b$.
Let $R(a,b)$ denote the closed rectangle with diagonal $ab$.

For a fixed pair of nodes $a,b \in V$, define a path
$\Pa_R(a \rightarrow b)$ as follows. Let $e \in V$ be the first node
along $\Pa(a \rightarrow b)$ that is not strictly interior to
$R(a,b)$. Then $\Pa_R(a \rightarrow b)$ is the subpath of $\Pa(a \rightarrow b)$
that extends between $a$ and $e$. In other words, $\Pa_R(a \rightarrow b)$ is
the path that follows the $Y_4$ edges pointing towards $b$,  truncated as soon
as it leaves the rectangle with diagonal $ab$, or as it reaches $b$. Formally,
\[
\Pa_R(a \rightarrow b) = \Pa(a \rightarrow b)[a, e]
\]
This definition is illustrated in Figure~\ref{fig:defs}b.

Our proofs will make use of the following two propositions.

\begin{proposition}
The sum of the lengths of crossing diagonals of a nondegenerate
(necessarily convex) quadrilateral $abcd$ is strictly greater than the sum of
the lengths of either pair of opposite sides:
\begin{eqnarray*}
|ac| + |bd| &>& |ab| + |cd| \\
|ac| + |bd| &>& |bc| + |da|
\end{eqnarray*}
\label{fact:quad}
\end{proposition}
%The proof breaks
This can be proved by partitioning
the diagonals into two pieces each at their intersection
point, and then applying the triangle inequality twice.

\begin{proposition}
For any triangle $\triangle abc$, the following inequalities hold:
\[|ac|^2
\begin{cases}
< |ab|^2 + |bc|^2, & \text{if } \angle{abc} < \pi/2 \\
= |ab|^2 + |bc|^2, & \text{if } \angle{abc} = \pi/2 \\
> |ab|^2 + |bc|^2, & \text{if } \angle{abc} > \pi/2 \\
\end{cases}
\]
\label{fact:tri}
\end{proposition}
This proposition follows immediately from the Law of Cosines applied
to triangle $\triangle abc$.

\begin{lemma}
For each pair of nodes $a, b \in V$,
\begin{equation}
|\Pa_R(a \rightarrow b)| \le |ab| \sqrt{2}
\label{eq:PR}
\end{equation}
Furthermore, each edge of $\Pa_R(a \rightarrow b)$ is no longer than $|ab|$.
\label{lem:PR}
\end{lemma}
\begin{proof}
Let $c$ be one of the two corners of $R(a, b)$, other than $a$ and $b$.
Let $\overrightarrow{de} \in \Pa_R(a \rightarrow b)$ be the last edge
on $\Pa_R(a \rightarrow b)$, which necessarily intersects
$\partial R(a,b)$ (note that it is possible that $e = b$). Refer to
Figure~\ref{fig:defs}b.
Then $|de| \le |db|$, otherwise $\overrightarrow{de}$ could not be in $Y_4$.
Since $db$ lies in the rectangle with diagonal $ab$, we have that
$|db| \le |ab|$,
and similarly for each edge on $\Pa_R(a \rightarrow b)$.
This establishes the latter claim of the lemma. For the first claim of
the lemma, let
\[
    p = \Pa_R(a \rightarrow b)[a,d] \oplus db
\]
Since $|de| \le |db|$, we have that $|\Pa_R(a \rightarrow b)| \le |p|$.
Since $p$ lies entirely inside $R(a,b)$ and consists of edges pointing
towards $b$, we have that $p$ is an $xy$-monotone path. It follows that
$|p| \le |ac| + |cb|$. We now show that
$|ac|+|cb| \le |ab|\sqrt{2}$, thus establishing the first
claim of the lemma.

Let $x=|ac|$ and $y=|cb|$. Then the inequality
$|ac|+|cb| \le |ab|\sqrt{2}$ can be written as
$x+y \leq \sqrt{2 x^2 + 2y^2}$, which is equivalent to
$(x-y)^2 \geq 0$. This latter inequality obviously holds,
completing the proof of the lemma.
\end{proof}

\begin{lemma}
Let $a, b, c, d \in V$ be four disjoint nodes such that
$\overrightarrow{ab}, \overrightarrow{cd} \in Y_4$, $b \in Q_i(a)$
and $d \in Q_i(c)$, for some $i \in \{1,2,3,4\}$. Then $ab$ and $cd$
cannot cross each other.
\label{lem:same.quadrant}
\end{lemma}
\begin{proof}
We may assume without loss of generality that $i=1$ and $c$ is to the
left of $a$. The proof is by contradiction. Assume that $ab$ and $cd$ cross each
other. Let $j$ be the intersection point between $ab$ and $cd$ (see Figure~\ref{fig:ratio}a).
Since $j \in Q_1(a) \cap Q_1(c)$,
it follows that $d \in Q_1(a)$ and $b \in Q_1(c)$. Thus $|ab| \le |ad|$, because
otherwise, $\overrightarrow{ab}$ cannot be in $Y_4$.
By
%% MS Fact --> Proposition
%% Fact~1
Proposition~\ref{fact:quad}
applied to the quadrilateral $adbc$,
\[
    |ad| + |cb| < |ab| +|cd|
\]
This along with the fact that $|ab| \le |ad|$ implies that
$|cb| < |cd|$, contradicting the fact that $\overrightarrow{cd} \in Y_4$.
\end{proof}

The next four lemmas (\ref{lem:quad}--\ref{lem:recross})
each concern a pair of crossing $Y_4$ edges,
culminating (in Lemma~\ref{lem:recross}) in the conclusion
that there is a short path in $Y_4$ between a pair of endpoints
of those edges.

\begin{lemma}
Let $a$, $b$, $c$ and $d$ be four disjoint nodes in $V$ such that
$\overrightarrow{ab}, \overrightarrow{cd} \in Y_4$, and
$ab$ crosses $cd$. Then the following are true: (i) the ratio between the
shortest side and the longer diagonal of the quadrilateral $acbd$ is
no greater than $1/\sqrt{2}$, and (ii) the
shortest side of the quadrilateral $acbd$ is strictly shorter
than either diagonal.
\label{lem:quad}
\end{lemma}
\begin{proof}
The first part of the lemma is a well-known fact that holds for
any quadrilateral (see~\cite{Quad51}, for instance).
For the second part of the lemma, let $ab$ be the shorter of the
diagonals of $acbd$, and assume without loss of generality
that $\overrightarrow{ab} \in Q_1(a)$.
Imagine two disks $D_a = D(a, |ab|)$ and $D_{b} = D(b,  |ab|)$,
as in Figure~\ref{fig:ratio}b. If either $c$ or $d$ belongs to
$D_a \cup D_{b}$,
then the lemma follows: a shortest quadrilateral edge is shorter
than $|ab|$.

%%%%%%%%%%%%%%%%%%%%%%%%%%%%%%%%%Figure Begin
\begin{figure}[htbp]
\centering
\includegraphics[width=0.85\linewidth]{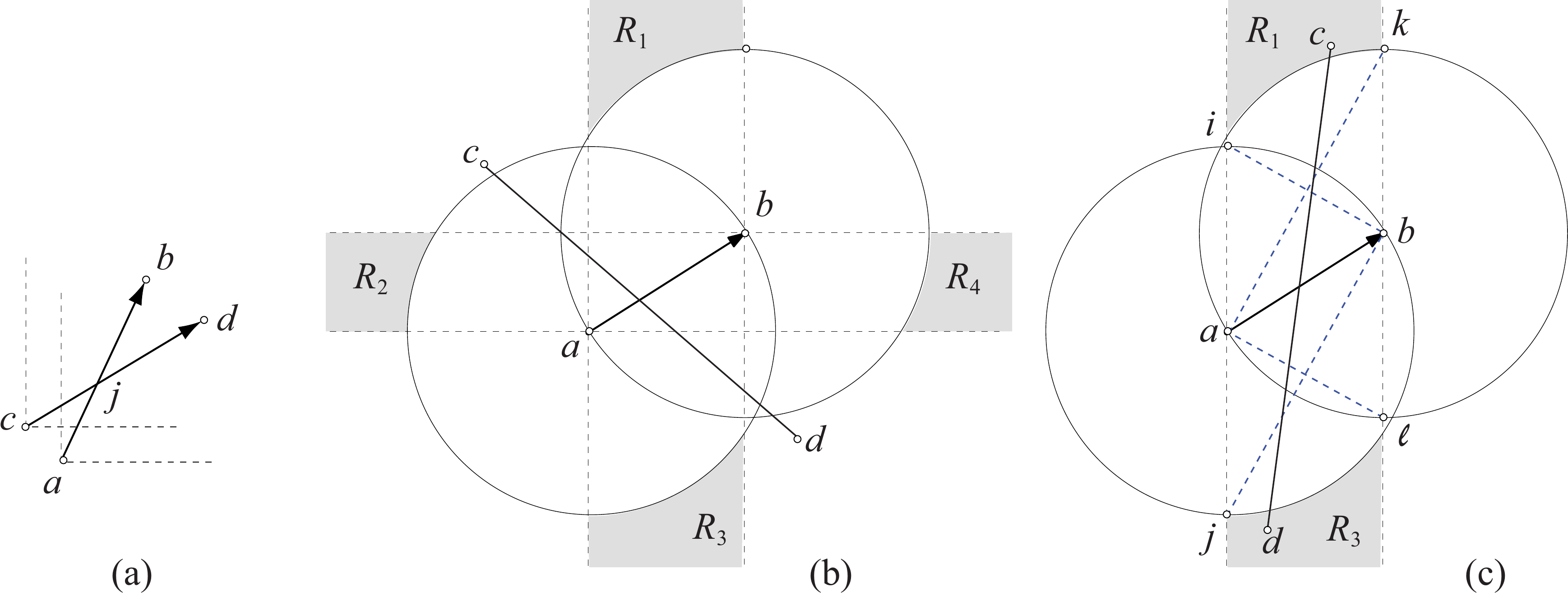}
\caption{(a) Lemma~\ref{lem:same.quadrant}. (b) Lemma~\ref{lem:quad}: $c \notin R_1\cup R_2\cup R_3 \cup R_4$ (c) Lemma~\ref{lem:quad}: $c \in R_1$.}
\label{fig:ratio}
\end{figure}
%%%%%%%%%%%%%%%%%%%%%%%%%%%%%%%%%Figure End

So suppose that neither $c$ nor $d$ lies in $D_a \cup D_{b}$. In
this case, we use the fact that $cd$ crosses $ab$ to show that
$\overrightarrow{cd}$ cannot be an edge in $Y_4$.
Define the following regions (see Figure~\ref{fig:ratio}b):
\begin{eqnarray*}
R_1&=& (Q_1(a) \cap Q_2(b)) \backslash( D_a \cup D_{b})\\
R_2&=& (Q_2(a) \cap Q_3(b)) \backslash( D_a \cup D_{b})\\
R_3&=& (Q_4(a) \cap Q_3(b)) \backslash( D_a \cup D_{b})\\
R_4&=& (Q_1(a) \cap Q_4(b)) \backslash( D_a \cup D_{b})
\end{eqnarray*}
If the node $c$ is not inside any of the regions $R_i$, for $i = \{1,2,3,4\}$, then
the nodes $a$ and $b$ are in the same quadrant of $c$ as $d$. In this case, note that
either $\angle cad > \pi/2$ or $\angle cbd > \pi/2$, which implies that either
$|ca|$ or $|cb|$ is strictly smaller than $|cd|$. These together show that
$\overrightarrow{cd} \notin Y_4$.

So assume that $c$ is in $R_i$ for some $i \in \{1,2,3,4\}$. In this situation,
the node $d$ must lie in the region $R_j$, with $j = (i+2) \bmod 4$ (with the
understanding that $R_0 = R_4$),
%[JOR: \emph{I don't think we want equivalence here.}]
because otherwise, (i) $a$ and $d$ are in the same quadrant of $c$ and
$|ca|<|cd|$ or (ii) $b$ and $d$ are in the same quadrant of $c$ and
$|cb|<|cd|$. Either case contradicts the fact $\overrightarrow{cd} \in Y_4$.
Consider now the case $c \in R_1$ and $d \in R_3$; the other cases are treated similarly.
Let $i$ and $j$ be the intersection points between $D_a$ and the
vertical line through $a$. Similarly, let $k$ and $\ell$ be the
intersection points between $D_b$ and the vertical line through $b$
(see Figure~\ref{fig:ratio}c).
Since $ij$ is a diameter of $D_a$, we have that $\angle{ibj} = \pi/2$ and similarly
$\angle{kal} = \pi/2$. Also note that $\angle{cbd} \ge \angle{ibj} = \pi/2$,
meaning that $|cd| > |cb|$. Similarly, $\angle{cad} \ge \angle{kal} = \pi/2$,
meaning that $|cd| > |ca|$. These along with the fact that at least one of $a$ and
$b$ is in the same quadrant for $c$ as $d$, imply that
$\overrightarrow{cd} \notin Y_4$. This completes the proof.
\end{proof}

\begin{lemma}
Let $a, b, c, d$ be four distinct nodes in $V$, with $c \in Q_1(a)$, such that
\begin{enumerate}
\item [(a)] $\overrightarrow{ab} \in Q_1(a)$ and
$\overrightarrow{cd} \in Q_2(c)$ are two edges in $Y_4$ that cross each other.
\item [(b)] $ad$ is a shortest side of the quadrilateral $acbd$. % < |aq|/\sqrt{2}.$
\end{enumerate}
Then $\Pa_R(a \rightarrow d)$ and $\Pa_R(d \rightarrow a)$ have a nonempty intersection.
\label{lem:ec1}
\end{lemma}
\begin{proof}
The proof consists of two parts showing that the following claims hold:
(i) $d \in Q_2(a)$ and (ii) $\Pa_R(d \rightarrow a)$ does not cross $ab$.

Before we prove these two claims, let us argue that they are sufficient
to prove the lemma. Lemma~\ref{lem:same.quadrant} and (i) imply that
$\Pa_R(a \rightarrow d)$ cannot cross $cd$. As a result,
$\Pa_R(a \rightarrow d)$ intersects the left side of the rectangle $R(d,a)$.
Consider the last edge $\overrightarrow{xy}$ of the path
$\Pa_R(d \rightarrow a)$. If this edge crosses the right side of $R(a,d)$,
then (ii) implies that $y$ is in the wedge bounded by $ab$ and the upwards
vertical ray starting at $a$; this implies that $|ay| < |ab|$, contradicting
the fact that $\overrightarrow{ab}$ is an edge in $Y_4$. Therefore,
$\overrightarrow{xy}$ intersects the bottom side of $R(d,a)$, and the lemma
follows (see Figure~\ref{fig:PRcross}b).

To prove the first claim (i), we observe that the assumptions in the lemma imply
that $d \in Q_1(a) \cup Q_2(a)$. Therefore, it suffices to prove that $d$
is not in $Q_1(a)$. Assume to the contrary that $d \in Q_1(a)$. Since
$c \in Q_1(a)$, it must be that $b \in Q_2(c)$; otherwise, $\angle{acb} \ge \pi/2$,
which implies $|ab| > |ac|$, contradicting the fact that
$\overrightarrow{ab} \in Y_4$. Let $i$ and $j$ be the
intersection points between $cd$ and $\partial D(a,|ab|)$, where $i$ is to
the left of $j$.
Since $\angle{dbc} \geq \angle{ibj} > \pi/2$, we have $|cb| < |cd|$. This,
together with the fact that $b$ and $d$ are in the same quadrant $Q_2(c)$,
contradicts the assumption that $\overrightarrow{cd}$ is an edge in $Y_4$.
This completes the proof of claim (i).

Next we prove claim (ii) by contradiction.
Thus, we assume that there is an edge $\overrightarrow{xy}$ on the path
$\Pa_R(d \rightarrow a)$ that crosses $ab$. Then necessarily
$x \in R(a, d)$ and $y \in Q_1(a) \cup Q_4(a)$.
If $y \in Q_4(a)$, then $\angle{xay} > \pi/2$, meaning that
$|xy| > |xa|$, a contradiction to the fact that $\overrightarrow{xy} \in Y_4$.
Thus, it must be that $y \in Q_1(a)$, as in Figure~\ref{fig:PRcross}a.
This implies that $|ab| \le |ay|$, because $\overrightarrow{ab} \in Y_4$.

%%%%%%%%%%%%%%%%%%%%%%%%%%%%%%%%%Figure Begin
\begin{figure}[htbp]
\centering
\includegraphics[width=0.65\linewidth]{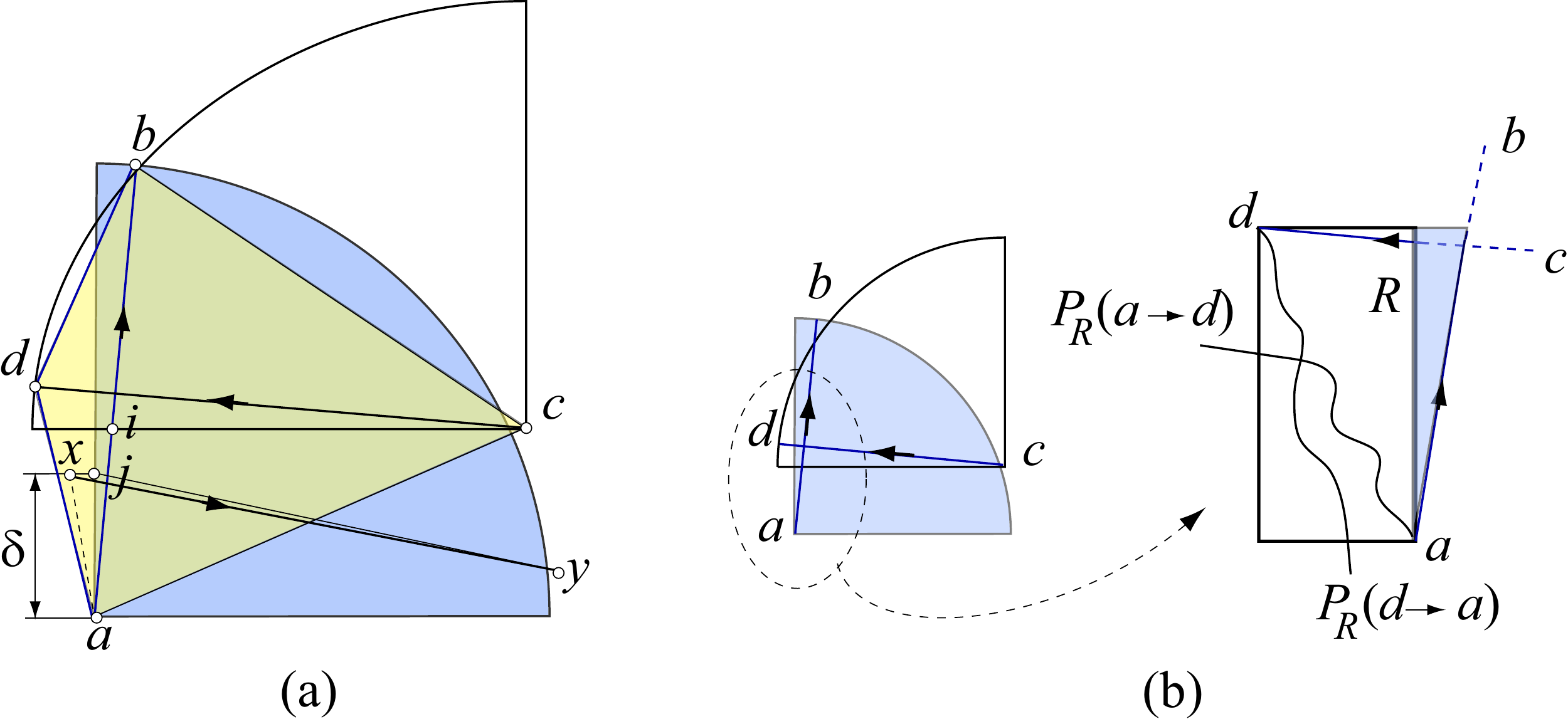}
\caption{(a) Lemma~\ref{lem:ec1}: $xy \in \Pa_R(d \rightarrow a)$ cannot cross $ab$.}
\label{fig:PRcross}
\end{figure}
%%%%%%%%%%%%%%%%%%%%%%%%%%%%%%%%%Figure End

The contradiction to our assumption that $\overrightarrow{xy}$ crosses $ab$
will be obtained by proving that $|xy| > |xa|$. Indeed, this inequality
contradicts the fact that $\overrightarrow{xy} \in Y_4$.

Let $\delta$ be the distance from $x$ to the horizontal line through $a$.
Our intermediate goal is to show that
\begin{equation}
\delta \le |ab|/\sqrt{2} .
\label{eq:d}
\end{equation}
We claim that $\angle{acb} < \pi/2$. Indeed, if this is not the case,
then $|ac|<|ab|$, contradicting the fact that $\overrightarrow{ab}$ is an
edge in $Y_4$. By a similar argument, and using the fact that
$\overrightarrow{cd}$ is an edge in $Y_4$, we obtain the inequality
$\angle{cbd} < \pi/2$.
We now consider two cases, depending on the relative lengths of $ac$ and $cb$.

\begin{enumerate}

\item

Assume first that $|ac| > |cb|$.
If $\angle{cad} \geq \pi/2$, then $|cd| \geq |ac| > |cb|$, contradicting
the fact that $\overrightarrow{cd}$ is an edge in $Y_4$ (recall that $b$
and $d$ are in the same quadrant of $c$). Therefore, we have
$\angle{cad} < \pi/2$.
Thus far we have established that three angles of the convex quadrilateral
$acbd$ are acute. It follows that the fourth one ($\angle{adb}$) is obtuse.
%% MS Fact --> Prop.
%% Fact~2
Proposition~\ref{fact:tri}
applied to $\triangle adb$ tells us that
\begin{equation*}
 |ab|^2  >  |ad|^2 + |db|^2 \geq 2|ad|^2 ,
\end{equation*}
where the latter inequality follows from the assumption that $ad$ is a
shortest side of $acbd$ (and, therefore, $|db| \geq |ad|$).
Thus, we have that $|ad| \le |ab|/\sqrt{2}$. This along with the fact that
$x \in R(a, d)$ implies inequality~(\ref{eq:d}).

\item
Assume now that $|ac| \le |cb|$. Let $i$ be the intersection
point between $ab$ and the horizontal line through $c$
(refer to Figure~\ref{fig:PRcross}a). Note that
$\angle{aic} \ge \pi/2$ and $\angle{bic} \le \pi/2$
(these two angles sum to $\pi$).
This along with
%% MS Fact --> Prop.
%% Fact~2
Proposition~\ref{fact:tri}
applied to triangle $\triangle aic$ shows that
\begin{equation*}
 |ac|^2  \ge  |ai|^2 + |ic|^2 .
\end{equation*}
Similarly,
%% MS Fact --> Prop.
%% Fact~2
Proposition~\ref{fact:tri}
applied to triangle $\triangle bic$ shows that
\begin{equation*}
|bc|^2 \le |bi|^2 + |ic|^2 .
\end{equation*}
The two inequalities above along with our assumption that $|ac| \le |cb|$
imply that $|ai| \le |bi|$, which in turn implies that $|ai| \le |ab|/2$,
because $|ai|+|ib|=|ab|$. Since $x$ is below $i$ (otherwise, $|cx| < |cd|$,
contradicting the fact that $\overrightarrow{cd}$ is an edge in $Y_4$),
we have $\delta \leq |ai|$. It follows that $\delta \leq |ab|/2$. %from which
%inequality~(\ref{eq:d}) follows.
\end{enumerate}

Finally we derive a contradiction using the now established inequality~(\ref{eq:d}).
Let $j$ be the orthogonal projection of $x$ onto the vertical line through
$a$ (thus $|aj| = \delta$).
Note that $\angle{ajy} < \pi/2$, because $y \in Q_4(x)$.
By
%% MS Fact --> Prop.
%% Fact~2
Proposition~\ref{fact:tri}
applied to $\triangle ajy$, we have
\begin{equation*}
 |ay|^2  <  |aj|^2 + |jy|^2 = \delta^2 + |jy|^2 .
\end{equation*}
Since $y$ and $b$ are in the same quadrant of $a$, and since
$\overrightarrow{ab} \in Y_4$, we have that $|ab| \le |ay|$.
This along with the inequality above and~(\ref{eq:d}) implies that
$|jy| \geq |ab|/\sqrt{2} \geq \delta$. By
%% MS Fact --> Prop.
%% Fact~2
Proposition~\ref{fact:tri}
applied to $\triangle xjy$,
we have
$|xy|^2 > |xj|^2 + |jy|^2 \geq |xj|^2 + \delta^2 = |xj|^2 + |ja|^2 = |xa|^2$. It
follows that $|xy| > |xa|$, contradicting our assumption that
$\overrightarrow{xy} \in Y_4$.
\end{proof}

\begin{lemma}
Let $a, b, c, d$ be four distinct nodes in $V$, with $c \in Q_1(a)$, such that
\begin{enumerate}
\item [(a)] $\overrightarrow{ab} \in Q_1(a)$ and $\overrightarrow{cd} \in Q_3(c)$ are two edges
in $Y_4$ that cross each other.
\item [(b)] $ad$ is a shortest side of the quadrilateral $acbd$.
\end{enumerate}
Then $\Pa_R(d \rightarrow a)$ does not cross $ab$.
%Then there is no edge $\overrightarrow{ab} \in Y_4$, with $a \in R(s, p)$, that crosses $pq$.
\label{lem:c1}
\end{lemma}
\begin{proof}
%%%%%%%%%%%%%%%%%%%%%%%%%%%%%%%%%Figure Begin
\begin{figure}[htbp]
\centering
\includegraphics[width=0.6\linewidth]{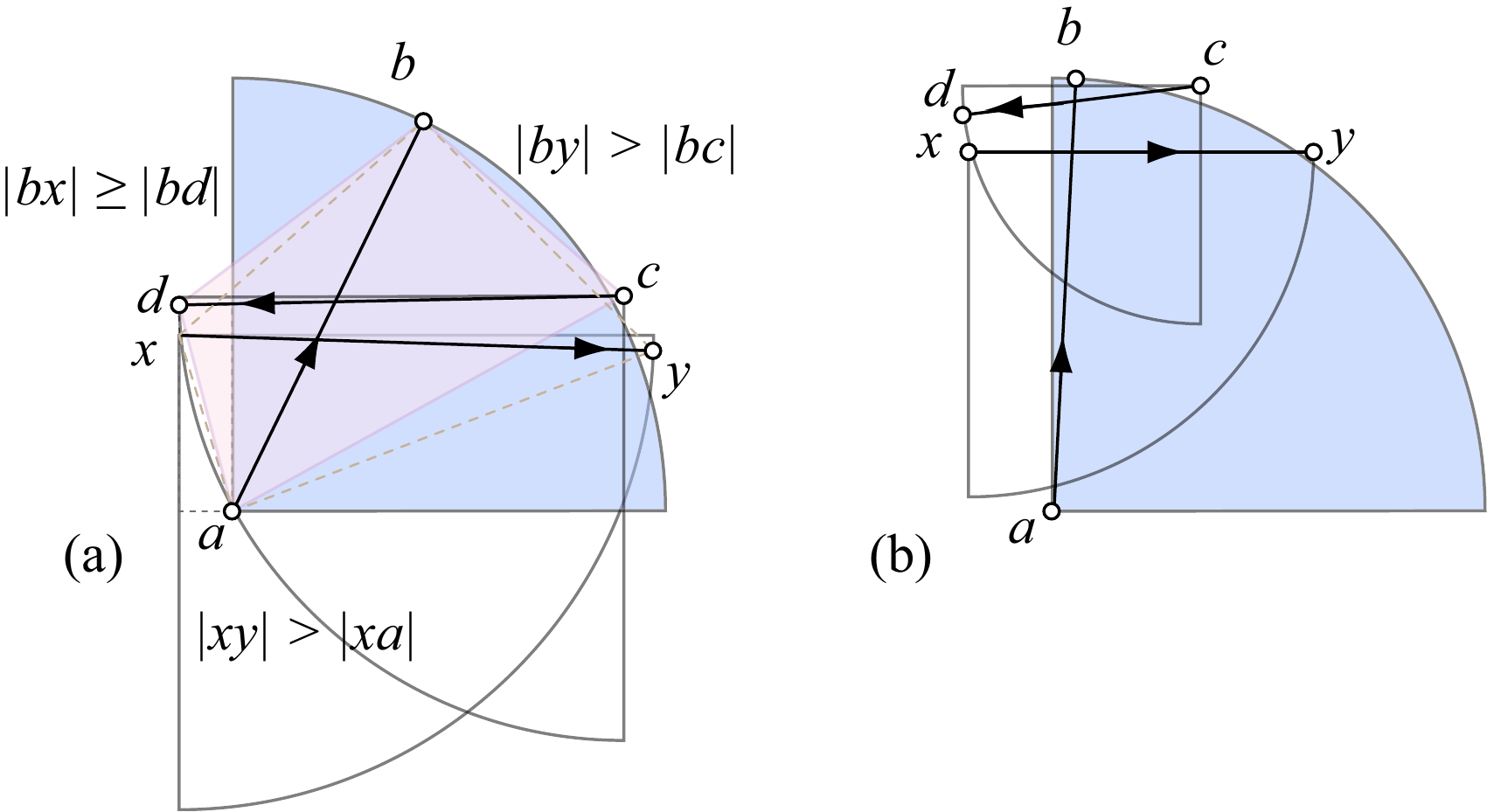}
\caption{Lemma~\ref{lem:c1}: (a)~$\Pa_R(d \rightarrow a)$ does not cross $ab$.
(b)~If $ad$ is not the shortest side of $acbd$, the lemma conclusion might not hold.}
\label{fig:c1}
\end{figure}
%%%%%%%%%%%%%%%%%%%%%%%%%%%%%%%%%Figure End
%
We first show that $d \notin Q_3(a)$. Assume the opposite. Since $c \in Q_1(a)$
and $d \in Q_3(a)$, we have that $\angle{cad} > \pi/2$. This implies that
$|ca| < |cd|$, which along with the fact that $a, d \in Q_3(c)$ contradict
the fact that $\overrightarrow{cd} \in Y_4$. Also note that $d \notin Q_1(a)$,
since in that case $ab$ and $cd$ could not intersect. In the following we
discuss the case $d \in Q_2(a)$; the case $d \in Q_4(a)$ is symmetric.

A first observation is that $c$ must lie below $b$; otherwise $|cb| < |cd|$
(since $\angle cbd > \pi/2$), which would contradict the fact that
$\overrightarrow{cd} \in Y_4$. We now prove by contradiction that there
is no edge in $\Pa_R(d \rightarrow a)$ crossing $ab$.
Assume the contrary, and let $\overrightarrow{xy} \in \Pa_R(d \rightarrow a)$
be such an edge. Then necessarily $x \in R(a, d)$ and
%By Lemma~\ref{lem:same.quadrant}, $\overrightarrow{ab} \not\in Q_1(a)$
%(since $\overrightarrow{pq} \in Q_1(p)$).
$\overrightarrow{xy} \in Q_4(x)$. Note that $y$ cannot lie below $a$;
otherwise $|xa| < |xy|$
(since $\angle xay > \pi/2$), which would contradict the fact that
$\overrightarrow{xy} \in Y_4$.
Also $y$ must lie outside $D(c, |cd|) \cap Q(c, d)$, otherwise
$\overrightarrow{cd}$ could not be in $Y_4$. These together show that $y$
sits to the right of $c$. See Figure~\ref{fig:c1}(a).
Then the following inequalities regarding the quadrilateral $xayb$
must hold:
\begin{itemize}
\item[(i)] $|by| > |bc|$, due to the fact that
$\angle{bcy} > \pi/2$.
\item[(ii)] $|bx| \ge |bd|$ ($|bx| = |bd|$ if $x$ and $d$ coincide).
If $x$ and $d$ are distinct, the inequality $|bx| > |bd|$ follows from
the fact that $|cx| \ge |cd|$ (since $x$ is outside $D(c, |cd|)$),
and
%% MS Fact --> Prop.
%% Fact~1
Proposition~\ref{fact:quad}
applied to the quadrilateral $xcbd$:
\[
  |bd| + |cx| < |bx| + |cd|
\]
\end{itemize}
Inequalities (i) and (ii) show that $by$ and $bx$ are longer
than sides of the quadrilateral $acbd$, and so they must be longer
than the shortest side of $acbd$, which by assumption~(b) of the lemma
is $ad$:
$\min\{|bx|, |by|\} \ge |ad| \ge |ax|$ (this latter inequality follows
from the fact that $x \in R(d, a)$). Also note that $|ab| \le |ay|$,
since $\overrightarrow{ab} \in Y_4$ and $y$ lies in the same quadrant
of $a$ as $b$. The fact that both diagonals of $xayb$ are in $Y_4$
enables us to apply Lemma~\ref{lem:quad}(ii) to conclude that
$ay$ is not a shortest side of the quadrilateral $xayb$. Thus
$xa$ is a shortest side of the quadrilateral $xayb$, and we can use
Lemma~\ref{lem:quad}(ii) to claim that
\[
|xa| < \min\{|xy|, |ab|\} \le |xy|.
\]
This contradicts our assumption that $\overrightarrow{xy} \in Y_4$.
\end{proof}

\noindent
Figure~\ref{fig:c1}(b) shows that the claim of the lemma might be false without
assumption~(b).
The next lemma relies on all of Lemmas~\ref{lem:PR}--\ref{lem:c1}.

\begin{lemma}
Let $a, b, c, d \in V$ be four distinct nodes such that
$\overrightarrow{ab} \in Y_4$ crosses
$\overrightarrow{cd} \in Y_4$, and let
$xy$ be a shortest side of the quadrilateral $abcd$.
Then there exist two paths $\Pa_x$ and $\Pa_y$ in $Y_4$, where
%% MS rephrased following line
%% $\Pa_x$ incident to $x$ and $\Pa_y$ incident to $y$,
$\Pa_x$ has $x$ as an endpoint and $\Pa_y$ has $y$ as an endpoint,
with the following properties:
\begin{enumerate}
\item[(a)] $\Pa_x$ and $\Pa_y$ have a nonempty intersection.
\item[(b)] $|\Pa_x| + |\Pa_y| \le 3 \sqrt{2} |xy|$.
\item[(c)] Each edge on $\Pa_x \cup \Pa_y$ is no longer than $|xy|$.
\end{enumerate}
\label{lem:pathcross}
\end{lemma}
\begin{proof}
Assume without loss of generality that $b \in Q_1(a)$.
We discuss the following exhaustive cases:
\begin{enumerate}
\item $c \in Q_1(a)$, and $d \in Q_1(c)$.
In this case, $ab$ and $cd$ cannot cross each other
(by Lemma~\ref{lem:same.quadrant}), so this case is finished.
%%%%%%%%%%%%%%%%%%%%%%%%%%%%%%%%%Figure Begin
\begin{figure}[htbp]
\centering
\includegraphics[width=0.8\linewidth]{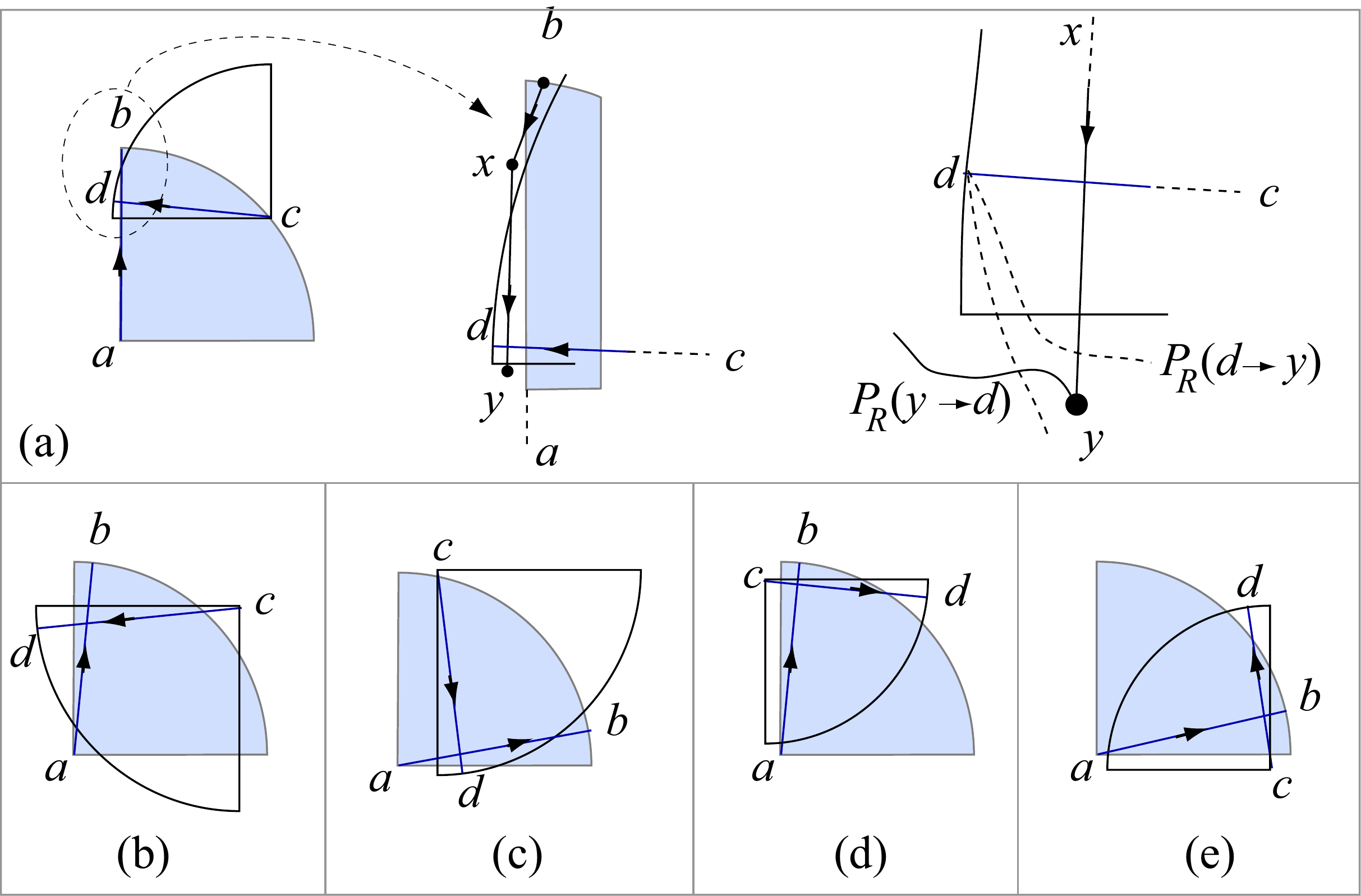}
\caption{Lemma~\ref{lem:pathcross}: (a) $c \in Q_1(a)$, and $d \in Q_2(c)$
(b)$c \in Q_1(a)$, and $d \in Q_3(c)$
(c) $c \in Q_2(a)$ (d) $c \in Q_4(a)$.}
\label{fig:pathcross}
\end{figure}
%%%%%%%%%%%%%%%%%%%%%%%%%%%%%%%%%Figure End
\item $c \in Q_1(a)$, and $d \in Q_2(c)$, as in
Figure~\ref{fig:pathcross}a. Since $ab$ crosses $cd$, $b \in Q_2(c)$.
Since $\overrightarrow{ab} \in Y_4$, $|ab| \le |ac|$.
Since $\overrightarrow{cd} \in Y_4$, $|cd| \le |cb|$.
These along with Lemma~\ref{lem:quad} imply that $ad$ and $db$ are the only
candidates for a shortest edge of $acbd$.

Assume first that $ad$ is a shortest edge of $acbd$.
By Lemma~\ref{lem:same.quadrant},
$\Pa_a = \Pa_R(a \rightarrow d)$ does not cross $cd$.
%[JOR: \emph{Check me here.  If I am not correct, then I don't see any
%citation/usage of Lemma~\ref{lem:c1}.}]
%% MS It is used in the third item below
%% MS modified the next few lines.
%%By Lemma~\ref{lem:c1},
%%$\Pa_d = \Pa_R(d \rightarrow a)$ does not cross $ab$. If follows that
%%$\Pa_a$ and $\Pa_d$ have a nonempty intersection.
It follows from
%% MS wrong lemma was cited
Lemma~\ref{lem:ec1}
that $\Pa_a$ and
$\Pa_d = \Pa_R(d \rightarrow a)$ have a nonempty intersection.
Furthermore, by Lemma~\ref{lem:PR},
$|\Pa_a| \le |ad| \sqrt{2}$ and
$|\Pa_d| \le |ad| \sqrt{2}$, and no edge on these paths
is longer than $|ad|$, proving the lemma
true for this case.

Consider now the case when $db$
is a shortest edge of $acbd$ (see Figure~\ref{fig:pathcross}a).
Note that $d$ is below $b$ (otherwise, $d \in Q_2(c)$ and
$|cd|>|cb|$) and, therefore, $b \in Q_1(d)$).
By Lemma~\ref{lem:same.quadrant}, $\Pa_d = \Pa_R(d \rightarrow b)$
does not cross $ab$. If $\Pa_b = \Pa_R(b \rightarrow d)$
does not cross $cd$, then $\Pa_b$ and $\Pa_d$ have a nonempty intersection,
proving the lemma true for this case.
Otherwise, there exists $\overrightarrow{xy} \in \Pa_R(b \rightarrow d)$
that crosses $cd$ (see Figure~\ref{fig:pathcross}a). Define
\begin{eqnarray*}
\Pa_b & = & \Pa_R(b \rightarrow d) \oplus \Pa_R(y \rightarrow d) \\
\Pa_d & = & \Pa_R(d \rightarrow y)
\end{eqnarray*}
By Lemma~\ref{lem:same.quadrant}, $\Pa_R(y \rightarrow d)$ does not cross $cd$.
Then $\Pa_b$ and $\Pa_d$ must have a nonempty intersection.
We now show that $\Pa_b$ and $\Pa_d$ satisfy conditions (b) and (c) of the lemma.
%
%% MS Fact --> Prop.
%% Fact~1
Proposition~\ref{fact:quad}
applied on the quadrilateral $xdyc$ tells us that
\[
   |xc| + |yd| < |xy| + |cd|
\]
We also have that $|cx| \ge |cd|$, since $\overrightarrow{cd} \in Y_4$ and $x$ is
in the same quadrant of $c$ as $d$. This along with the inequality above implies
$|yd| < |xy|$. Because $xy \in \Pa_R(b \rightarrow d)$, by Lemma~\ref{lem:PR} we
have that $|xy| \le |bd|$, which along with the previous inequality shows that
$|yd| < |bd|$.
This along with Lemma~\ref{lem:PR} shows that condition (c) of the lemma
is satisfied. Furthermore,
$|\Pa_R(y \rightarrow d)| \le |yd|\sqrt{2}$ and
$|\Pa_R(d \rightarrow y)| \le |yd|\sqrt{2}$.
It follows that $|\Pa_b| + |\Pa_d| \le 3\sqrt{2}|bd|$.

\item $c \in Q_1(a)$, and $d \in Q_3(c)$, as in Figure~\ref{fig:pathcross}b.
Then $|ac| \ge \max\{ab, cd\}$, and by Lemma~\ref{lem:quad} $ac$ is not
a shortest edge of $acbd$.
The case when $bd$ is a shortest edge of $acbd$ is settled by
Lemmas~\ref{lem:same.quadrant} and~\ref{lem:PR}: Lemma~\ref{lem:same.quadrant}
tells us that $\Pa_d = \Pa_R(d \rightarrow b)$ does not cross $ab$, and
$\Pa_b = \Pa_R(b \rightarrow d)$ does not cross $cd$. It follows that $\Pa_d$ and
$\Pa_b$ have a nonempty intersection. Furthermore, Lemma ~\ref{lem:PR} guarantees
that $\Pa_d$ and $\Pa_b$ satisfy conditions (b) and (c) of the lemma.

Consider now the case when $ad$ is a shortest edge of $acbd$; the case
when $bc$ is shortest is symmetric.
%By Lemma~\ref{lem:PR}, $|sa| \le |pq|/\sqrt{2}$.
By
%% MS wrong lemma was cited
%% Lemma~\ref{lem:ec1},
Lemma~\ref{lem:c1},
$\Pa_R(d \rightarrow a)$ does not cross $ab$.
If $\Pa_R(a \rightarrow d)$ does not cross $cd$, then this case is settled:
$\Pa_d = \Pa_R(d \rightarrow a)$ and $\Pa_a = \Pa_R(a \rightarrow d)$ satisfy the three conditions
of the lemma. Otherwise, let $\overrightarrow{xy} \in \Pa_R(a \rightarrow d)$
be the edge crossing $cd$. Arguments similar to the ones used in case 1 above show that
\begin{eqnarray*}
\Pa_a & = & \Pa_R(a \rightarrow d) \oplus \Pa_R(y \rightarrow d) \\
\Pa_d & = & \Pa_R(d \rightarrow y)
\end{eqnarray*}
are two paths that satisfy the conditions of the lemma.

\item $c \in Q_1(a)$, and $d \in Q_4(c)$, as in Figure~\ref{fig:pathcross}c.
Note that a horizontal reflection of Figure~\ref{fig:pathcross}c, followed
by a rotation of $\pi/2$, depicts a case identical to case 1, which has
already been settled.

\item $c \in Q_2(a)$, as in Figure~\ref{fig:pathcross}d.
Note that Figure~\ref{fig:pathcross}d rotated by $\pi/2$ depicts
a case identical to case 1, which has already been settled.

\item $c \in Q_3(a)$. Then it must be that $d \in Q_1(c)$, otherwise
$cd$ cannot cross $ab$. By Lemma~\ref{lem:same.quadrant} however,
$ab$ and $cd$ may not cross, unless one of them is not in $Y_4$.

\item $c \in Q_4(a)$, as in Figure~\ref{fig:pathcross}e.
Note that a vertical reflection of Figure~\ref{fig:pathcross}e depicts
a case identical to case 1, so this case is settled as well.
\end{enumerate}
Having exhausted all cases, we conclude that the lemma holds.
\end{proof}

\noindent
We are now ready to establish the main lemma of this section, showing
that there is a short path between the endpoints of two intersecting
edges in $Y_4$.

\begin{lemma}
Let $a, b, c, d \in V$ be four distinct nodes such that
$\overrightarrow{ab} \in Y_4$ crosses
$\overrightarrow{cd} \in Y_4$, and let
$xy$ be a shortest side of the quadrilateral $abcd$.
Then $Y_4$ contains a path $p(x,y)$ connecting $x$ and $y$, of length
\[
    |p(x,y)| \le \frac{6}{\sqrt{2}-1} \cdot |xy|.
\]
Furthermore, no edge on $p(x,y)$ is longer than $|xy|$.
\label{lem:recross}
\end{lemma}
\begin{proof}
Let $\Pa_x$ and $\Pa_y$ be the two paths whose existence in $Y_4$ is
guaranteed by Lemma~\ref{lem:pathcross}. By condition (c) of
Lemma~\ref{lem:pathcross}, no edge on $\Pa_x$ and $\Pa_y$ is longer
than $|xy|$. By condition (a) of Lemma~\ref{lem:pathcross}, $\Pa_x$ and
$\Pa_y$ have a nonempty intersection. If $\Pa_x$ and $\Pa_y$
share a node $u \in V$, then the path
\[
p(x, y) = \Pa_x[x, u] \oplus \Pa_y[y,u]
\]
is a path from $x$ to $y$ in $Y_4$ no longer than $3\sqrt{2}|xy|$;
the length restriction follows from guarantee (b) of Lemma~\ref{lem:pathcross}.
Otherwise, let $\overrightarrow{a'b'} \in \Pa_x$ and
$\overrightarrow{c'd'} \in \Pa_y$ be two edges crossing
each other. Let $x'y'$ be a shortest side of the quadrilateral
$a'c'b'd'$, with $x' \in \Pa_x$ and $y' \in \Pa_y$.
Lemma~\ref{lem:pathcross} tells us that $|a'b'| \le |xy|$ and $|c'd'| \le |xy|$.
These along with Lemma~\ref{lem:quad} imply that
\begin{equation}
|x'y'| \le |xy| / \sqrt{2}.
\label{eq:rq}
\end{equation}
This enables us to derive a recursive formula for computing a path
$p(x,y) \in Y_4$ as follows:
\begin{equation}
p(x, y) =
\begin{cases}
x, & \text{if $x = y$} \\
\Pa_x[x,x'] \oplus \Pa_y[y, y'] \oplus p(x', y'), & \text{if $x \neq y$}
\label{eq:pab}
\end{cases}
\end{equation}
%Let $|ab| = d$.
Next we use induction on the length of $xy$ to prove the claim of the lemma.
The base case corresponds to $x=y$, case in which $p(x,y)$ degenerates to a point
and $|p(x,y)| = 0$.
To prove the inductive step, pick a shortest side $xy$ of a quadrilateral $acbd$, with
$\overrightarrow{ab}, \overrightarrow{cd} \in Y_4$ crossing each other, and assume that
the lemma holds for all such sides shorter than $xy$.
Let $p(x,y)$ be the path determined recursively as in~(\ref{eq:pab}).
By the inductive hypothesis, we have that $p(x',y')$ contains no edges longer
than $|x'y'| \le |xy|$, and
\[
    |p(x',y')| \le \frac{6}{\sqrt{2}-1} |x'y'| \le \frac{6}{2 - \sqrt{2}}|xy| .
\]
This latter inequality follows from~(\ref{eq:rq}). This along with Lemma~\ref{lem:pathcross}
and formula~(\ref{eq:pab}) implies
\[
    |p(x,y)| \le (3\sqrt{2} + \frac{6}{2-\sqrt{2}})\cdot|xy| = \frac{6}{\sqrt{2}-1}\cdot|xy|.
\]
This completes the proof.
\end{proof}

\section{$Y^\infty_4$ and $Y_4$}
\label{sec:Yinf+Y4}
We prove that every individual edge of $Y^\infty_4$ is spanned by a
short path in $Y_4$. This, along with the result of
Theorem~\ref{thm:y4infspanner}, establishes that $Y_4$ is a spanner.

Fix an edge $\overrightarrow{ab} \in Y_4^\infty$. Define an edge or
a path as \emph{short} if its length is within a constant factor
of $|ab|$. In our proof that $ab$ is spanned by a short path in
$Y_4$, we will make use of the following three statements
%% MS added following line
(which will be proved in the appendix).
\begin{description}
\item[S1]
If $ab$ is short, then $\Pa_R(a \rightarrow b)$,
and therefore its reverse, $\Pa_R^{-1}(a \rightarrow b)$,
are short by Lemma~\ref{lem:PR}.
\item[S2]
If $ab \in Y_4$ and $cd \in Y_4$ are short, and if
$ab$ intersects $cd$, Lemma~\ref{lem:recross} shows that
then there is a short path between any two of the endpoints
of these edges.
\item[S3]
If $p(a,b)$ and $p(c,d)$ are short paths that intersect,
then there is a short path $P$ between any two of the endpoints
of these paths, by {\bf S2}.
\end{description}

\begin{lemma}
For any edge $ab \in Y_4^\infty$, there is a short path $p(a,b) \in Y_4$ of length
\[|p(a,b)| \le (29 + 23\sqrt{2})|ab|.\]
\label{lem:lool2short}
\end{lemma}
\begin{proof}
For the sake of clarity, we only prove here that there is a short path $p(a,b)$, and defer
the calculations of the actual stretch factor of $p(a,b)$ to the appendix.
Assume without loss of generality that $\overrightarrow{ab} \in Y_4^\infty$, and
$\overrightarrow{ab} \in Q_1(a)$. If $\overrightarrow{ab} \in Y_4$, then $p(a,b) = ab$
and the proof is finished. So assume the opposite, and let
$\overrightarrow{ac} \in Q_1(a)$ be the edge in $Y_4$; since $Q_1(a)$ is nonempty,
$\overrightarrow{ac}$ exists.
Because $\overrightarrow{ac} \in Y_4$ and $b$ is in the same quadrant of $a$
as $c$, we have that
\begin{eqnarray}
\nonumber |ac| & \le & |ab| ~~~~~~~~~~~~~\mbox{(i)}\\
          |bc| & \le & |ac|\sqrt{2}~~~~~~~~~\mbox{(ii)}
\label{eq:o1}
\end{eqnarray}
Thus both $ac$ and $bc$ are short.
And this in turn implies that $\Pa_R(b \rightarrow c)$ is short by {\bf S1}.
We next focus on $\Pa_R(b \rightarrow c)$.
Let $b'\notin R(b,c)$ be the other endpoint of $\Pa_R(b \rightarrow c)$.
We distinguish three cases.

%%%%%%%%%%%%%%%%%%%%%%%%%%%%%%%%%Figure Begin
\begin{figure}[hp]
\centering
\includegraphics[width=0.8\linewidth]{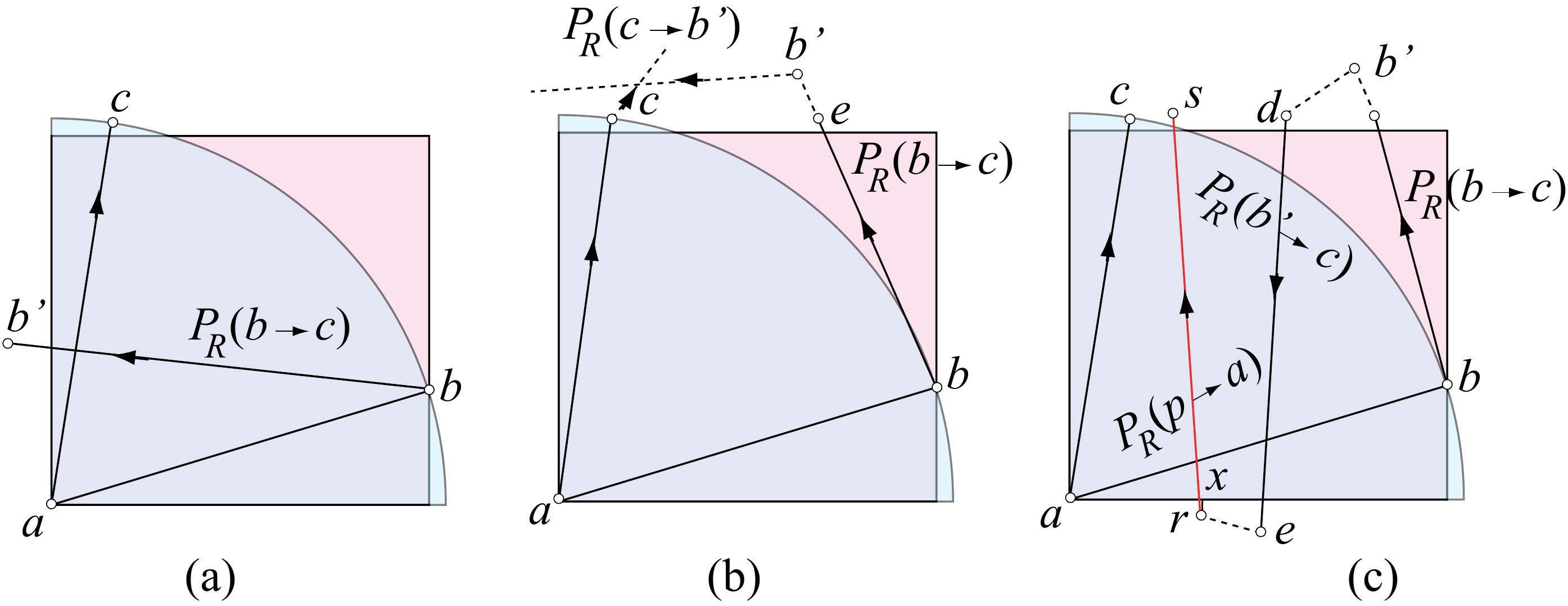}
\caption{Lemma~\ref{lem:lool2short}: (a) Case 1: $\Pa_R(b \rightarrow c)$ and $ac$ have a nonempty intersection.
(b) Case 2: $\Pa_R(b' \rightarrow a)$ and $ab$ have an empty intersection.
(c) Case 3: $\Pa_R(b' \rightarrow a)$ and $ab$ have a non-empty intersection.}
\label{fig:lool2}
\end{figure}
%%%%%%%%%%%%%%%%%%%%%%%%%%%%%%%%%Figure End

\noindent
{\bf Case 1:} $\Pa_R(b \rightarrow c)$ and $ac$ intersect.
Then by {\bf S3} there is a short path $p(a,b)$ between $a$ and $b$.

\medskip
\noindent
{\bf Case 2:} $\Pa_R(b \rightarrow c)$ and $ac$ do not intersect,
and $\Pa_R(b' \rightarrow a)$ and $ab$ do not intersect (see Figure~\ref{fig:lool2}b).
Note that because $b'$ is
the endpoint of the short path $\Pa_R(b \rightarrow c)$, the triangle inequality on $\triangle abb'$
implies that $ab'$ is short, and therefore
$\Pa_R(b' \rightarrow a)$ is short. We consider two cases:
\begin{enumerate}
\item[(i)] $\Pa_R(b' \rightarrow a)$ intersects $ac$.
Then by {\bf S3} there is a short path $p(a,b')$.
So
$$
p(a,b) = p(a,b') \oplus \Pa_R^{-1}(b \rightarrow c)
$$
is short.

\item[(ii)] $\Pa_R(b' \rightarrow a)$ does not intersect $ac$.
Then $\Pa_R(c \rightarrow b')$ must intersect
$\Pa_R(b \rightarrow c) \oplus \Pa_R(b' \rightarrow a)$.
Next we establish that $b'c$ is short. Let $\overrightarrow{eb'}$ be the last edge
of $\Pa_R(b \rightarrow c)$, and so incident to $b'$
(note that $e$ and $b$ may coincide). Because $\Pa_R(b \rightarrow c)$ does not intersect
$ac$, $b'$ and $c$ are in the same quadrant for $e$. It follows that
$|eb'| \le |ec|$ and $\angle{b'ec} < \pi/2$.
These along with
%% MS Fact --> Prop.
%% Fact~(\ref{fact:tri})
Proposition~\ref{fact:tri}
for $\triangle b'ec$ imply that
$|b'c|^2 < |b'e|^2 + |ec|^2 \le 2|ec|^2 < 2|bc|^2$ (this latter inequality
uses the fact that $\angle{bec} > \pi/2$, which implies that $|ec| < |bc|$).
It follows that
\begin{eqnarray}
 |b'c| & \le & |bc| \sqrt{2} ~\le ~2|ac| ~~~~~~~~~~~~~\mbox{(by~(\ref{eq:o1})ii)}
\label{eq:o4a}
\end{eqnarray}
Thus $b'c$ is short, and by {\bf S1} we have that $\Pa_R(c \rightarrow b')$ is short.
Since $\Pa_R(c \rightarrow b')$ intersects the short path
$\Pa_R(b \rightarrow c) \oplus \Pa_R(b' \rightarrow a)$,
there is by {\bf S3} a short path $p(c,b)$, and so
$$
p(a,b) = ac \oplus p(c,b)
$$
is short.

\end{enumerate}

\medskip
\noindent
{\bf Case 3:} $\Pa_R(b \rightarrow c)$ and $ac$ do not intersect,
and $\Pa_R(b' \rightarrow a)$ intersects $ab$ (see Figure~\ref{fig:lool2}c).
If $\Pa_R(b' \rightarrow a)$ intersects $ab$ at $a$, then
$p(a,b) = \Pa_R(b \rightarrow c) \oplus \Pa_R(b' \rightarrow a)$ is short.
So assume otherwise, in which case
there is an edge $\overrightarrow{de} \in \Pa_R(b' \rightarrow a)$ that
crosses $ab$.
Then $d \in Q_1(a)$, $e \in Q_3(a) \cup Q_4(a)$, and $e$ and $a$ are in
the same quadrant for $d$. Note however that $e$ cannot lie in
$Q_3(a)$, since in that case $\angle{dae} > \pi/2$, which would imply
$|de| > |da|$, which in turn would imply $\overrightarrow{de} \notin Y_4$.
So it must be that $e \in Q_4(a)$.

Next we show that $\Pa_R(e \rightarrow a)$ does not cross $ab$. Assume the opposite,
and let $\overrightarrow{rs} \in \Pa_R(e \rightarrow a)$ cross $ab$. Then
$r \in Q_4(a)$, $s \in Q_1(a) \cup Q_2(a)$, and $s$ and $a$ are in
the same quadrant for $r$. Arguments similar to the ones above show that
$s \notin Q_2(a)$, so $s$ must lie in $Q_1(a)$. Let $d$ be the $L_\infty$
distance from $a$ to $b$. Let $x$ be the projection of $r$ on the horizontal
line through $a$. Then
\[
|rs| \ge |rx| + d \ge |rx| + |xa| > |ra| ~~~~\mbox{(by the triangle inequality)}
\]
Because $a$ and $s$ are in the same quadrant for $r$, the inequality above
contradicts $\overrightarrow{rs} \in Y_4$.

We have established that $\Pa_R(e \rightarrow a)$ does not cross $ab$. Then
$\Pa_R(a \rightarrow e)$ must intersect $\Pa_R(e \rightarrow a) \oplus de$.
Note that $de$ is short because it is in the short path $\Pa_R(b' \rightarrow a)$.
Thus $ae$ is short, and so $\Pa_R(a \rightarrow e)$ and $\Pa_R(e \rightarrow a)$
are short.
Thus we have two intersecting short paths, and so by {\bf S3}
there is a short path $p(a,e)$.
Then
$$
p(a,b) = p(a,e) \oplus \Pa_R^{-1}(b' \rightarrow a) \oplus  \Pa_R^{-1}(b \rightarrow c)
$$
is short. Calculations deferred to the appendix show that, in each of these cases, the
stretch factor for $p(a,b)$ does not exceed $29+23\sqrt{2}$.
\end{proof}

\noindent
Our main result follows immediately from Theorem~\ref{thm:y4infspanner} and Lemma~\ref{lem:lool2short}:

\begin{theorem}
$Y_4$ is a $t$-spanner, for $t \ge 8(29+23\sqrt{2})$.
\end{theorem}

\section{Conclusion}
Our results settle a long-standing open problem, asking whether $Y_4$
is a spanner or not. We answer this question positively, and
establish a loose stretch factor of $8(29+23\sqrt{2})$.
Experimental results, however, indicate a stretch factor of the
%% MS who did these experiments?
%% [MD:] Joe did these experiments.
%% Are the experiments for points uniformly chosen in a square or circle?
order $1 + \sqrt{2}$, a factor of 200 smaller.
%[JOR: \emph{492 vs. 2.4}.]
Finding tighter stretch factors for both $Y_4^\infty$ and $Y_4$
remain interesting open problems. Establishing
whether $Y_5$ and $Y_6$ are spanners or not is also open. 

\noindent
%[JOR: \emph{I think it would be useful here to summarize which
%$Y_k$ remain unresolved as spanners or not.
%$\{Y_5, Y_6\}$?
%% MS Joe is correct: Y_5 and Y_6 are still open.
%That's implied by the Introduction...}]

%\small
%\bibliographystyle{plain}
%\bibliography{../../spannerbib}
\def\cprime{$'$}

\section{Appendix}
\label{sec:appendix}
\subsection{Calculations for the stretch factor of $p(a,b)$ in Lemma~\ref{lem:lool2short}}
We start by computing the stretch factor of the short paths claimed by statements {\bf S2}
and {\bf S3}.

\begin{description}
\item[S2]
If $ab \in Y_4$ and $cd \in Y_4$ are short, and if
$ab$ intersects $cd$, %Lemma~\ref{lem:recross} shows that
then there is a short path $P$ between any two of the endpoints
of these edges, of length % The length of $P$ is bounded above by
\begin{equation}
|P| \le |ab| + |cd| + 3(2+\sqrt{2})\max\{|ab|, |cd|\}
\label{eq:s2}
\end{equation}
This upper bound can be derived as follows. Let $xy$ be a shortest side of
the quadrilateral $acbd$. By Lemma~\ref{lem:recross}, $Y_4$
contains a path $p(x, y)$ no longer than $6(\sqrt{2}+1)|xy|$.
By Lemma~\ref{lem:quad}, $|xy| \le \max\{|ab|, |cd|\}/\sqrt{2}$.
These together with the fact that
$|P| \le |ab| +|cd| +|p(x,y)|$ yield inequality~(\ref{eq:s2}).
\item[S3]
If $p(a,b)$ and $p(c,d)$ are short paths that intersect,
then there is a short path $P$ between any two of the endpoints
of these paths, of length %
\begin{equation}
|P| \le |p(a,b)| + |p(c,d)| + 3(2+\sqrt{2})\max\{|ab|, |cd|\}
\label{eq:s3}
\end{equation}
This follows immediately from {\bf S2} and the fact that
no edge on $p(a,b) \cup p(c,d)$ is longer than $\max\{|ab|, |cd|\}$ (by Lemma~\ref{lem:recross}).
\end{description}

\medskip
\noindent
{\bf Case 1:} $\Pa_R(b \rightarrow c)$ and $ac$ intersect. Then by {\bf S3} we have

\begin{eqnarray*}
\nonumber |p(a,b)| & \le & |\Pa_R(b,c)| + |ac| + 3(2+\sqrt{2}) \max\{|bc|, |ac|\} \\
                   & \le & \sqrt{2}|bc| + |ac| + 3(2+\sqrt{2})\sqrt{2}|ac| ~~~~~~~~~~~~~~~~\mbox{(by~(\ref{eq:PR}),~(\ref{eq:o1})ii)} \\
                   & = & 3(3+2\sqrt{2})|ac| \le 3(3+2\sqrt{2})|ab| ~~~~~~~~~~~~~~~~\mbox{(by~({\ref{eq:o1}})i)} \\
\label{eq:o2}
\end{eqnarray*}

\medskip
\noindent
{\bf Case 2(i):}
$\Pa_R(b \rightarrow c)$ and $ac$ do not intersect;
$\Pa_R(b' \rightarrow a)$ and $ab$ do not intersect;
and $\Pa_R(b' \rightarrow a)$ intersects $ac$.
By {\bf S3}, there is a short path $p(a,b')$ of length
\begin{eqnarray}
\nonumber |p(a,b')| & \le & |\Pa_R(b',a)| + |ac| + 3(2+\sqrt{2}) \max\{|b'a|, |ac|\} \\
           & \le & |b'a|\sqrt{2} + |ac| + 3(2+\sqrt{2}) \max\{|b'a|, |ac|\} ~~~~~~~~~\mbox{(by~(\ref{eq:PR}))}
\label{eq:o3}
\end{eqnarray}
Next we establish an upper bound on $|b'a|$. By the triangle inequality,
\begin{eqnarray}
          |ab'| & < & |ac| + |cb'| ~\le ~3|ac| ~~~~~~~~\mbox{(by~(\ref{eq:o4a}))}
\label{eq:o4b}
\end{eqnarray}
Substituting this inequality in~(\ref{eq:o3}) yields
\begin{eqnarray}
|p(a,b')| & \le & (19+12\sqrt{2})|ac|
\label{eq:o5a}
\end{eqnarray}
Thus $p(a,b) = p(a,b') \oplus \Pa_R^{-1}(b \rightarrow c)$ is a path in $Y_4$ of length
\begin{eqnarray*}
        |p(a,b)| & \le & |p(a,b')| + |bc| \sqrt{2} ~~~~~~~~~~~~~\mbox{(by~(\ref{eq:PR}))}\\
                   & \le & |p(a,b')| + 2|ac|  ~~~~~~~~~~~~~~~\mbox{(by~(\ref{eq:o1})ii)} \\
                   & \le & (21+12\sqrt{2})|ac| ~~~~~~~~~~~~~~~\mbox{(by~(\ref{eq:o5a}))} \\
                   & \le & (21+12\sqrt{2})|ab| ~~~~~~~~~~~~~~~\mbox{(by~(\ref{eq:o1})i)} \\
\end{eqnarray*}

\medskip
\noindent
{\bf Case 2(ii):}
$\Pa_R(b \rightarrow c)$ and $ac$ do not intersect;
$\Pa_R(b' \rightarrow a)$ and $ab$ do not intersect;
and $\Pa_R(b' \rightarrow a)$ does not intersect $ac$.
Then $\Pa_R(c \rightarrow b')$ must intersect
$\Pa_R(b \rightarrow c) \oplus \Pa_R(b' \rightarrow a)$.
By {\bf S3} there is a short path $p(c,b)$ of length
\begin{eqnarray*}
\nonumber |p(c,b)| & \le & |\Pa_R(c \rightarrow b')| + |\Pa_R(b \rightarrow c)| + |\Pa_R(b' \rightarrow a)| +  3(2+\sqrt{2}) \max\{|cb'|, |bc|, |b'a|\} \\
                   & \le & (|cb'| + |bc| + |b'a|)\sqrt{2} + 3(2+\sqrt{2}) \max\{|cb'|, |bc|, |b'a|\} ~~~~~~~\mbox{(by~(\ref{eq:PR}))}
\label{eq:o6}
\end{eqnarray*}
Inequalities~(\ref{eq:o1})ii, ~(\ref{eq:o4a}) and~(\ref{eq:o4b}) imply that
$\max\{|cb'|, |bc|, |b'a|\} \le 3ac$. Substituting in the above,
we get
\begin{eqnarray*}
\nonumber |p(c,b)| & \le & (2+\sqrt{2}+3)\sqrt{2}|ac| +  9(2+\sqrt{2})|ac| \\
                   & \le & (20 + 14\sqrt{2}) |ac| ~~~~~~~~~~~~~~~\mbox{(by~(\ref{eq:o1})i)}
\label{eq:o7}
\end{eqnarray*}
Thus $p(a,b) = ac \oplus p(c,b)$ is a path in $Y_4$ from $a$ to $b$ of length
\begin{eqnarray*}
\nonumber |p(a,b)| & \le & (21 + 14\sqrt{2}) |ac|  ~\le ~(21 + 14\sqrt{2}) |ab|~~~~~~~~\mbox{(by~(\ref{eq:o1})i)} \\
\label{eq:o8}
\end{eqnarray*}

\medskip
\noindent
{\bf Case 3:} $\Pa_R(b \rightarrow c)$ and $ac$ do not intersect,
and $\Pa_R(b' \rightarrow a)$ intersects $ab$.
If $\Pa_R(b' \rightarrow a)$ intersects $ab$ at $a$, then
$p(a,b) = \Pa_R(b \rightarrow c) \oplus \Pa_R(b' \rightarrow a)$ is clearly short and
does not exceed the spanning ratio of the lemma.
Otherwise, there is an edge $\overrightarrow{de} \in \Pa_R(b' \rightarrow a)$
that crosses $ab$, and
$\Pa_R(a \rightarrow e)$ intersects $\Pa_R(e \rightarrow a) \oplus de$
(as established in the proof of Lemma~\ref{lem:lool2short}).
By {\bf S3} there is a short path $p(a,e)$ of length
\begin{eqnarray}
\nonumber |p(a,e)| & \le & |\Pa_R(a \rightarrow e)| + |\Pa_R(e \rightarrow a)| + |de| +  3(2+\sqrt{2}) \max\{|ae|, |de|\} \\
                   & \le & 2|ae|\sqrt{2} + |de| + 3(2+\sqrt{2}) \max\{|ae|, |de|\} ~~~~~~~~~~~~~\mbox{(by~(\ref{eq:PR}))}
\label{eq:o9}
\end{eqnarray}
A loose upper bound on $|ae|$ can be obtained by employing
%% Ms Fact --> Prop.
%% Fact~1
Proposition~\ref{fact:quad}
to the quadrilateral
$aebd$: $|ae| + |bd| < |ab| + |de| < |ab| + |ab'|$. Substituting the upper bound for
$ab'$ from~(\ref{eq:o4b}) yields
\begin{equation}
|ae| < |ab| + 3|ac| \le 4|ab|
\label{eq:ae}
\end{equation}
By Lemma~\ref{lem:PR}, $|de| \le |ab'|$ (since $de \in \Pa_R(b' \rightarrow a)$), which
along with~(\ref{eq:o4b}) implies
\begin{equation}
|de| \le 3|ab|
\label{eq:de}
\end{equation}
Substituting~(\ref{eq:ae}) and~(\ref{eq:de}) in ~(\ref{eq:o9}) yields
\begin{eqnarray}
\nonumber |p(a,e)| & \le & (27 + 20\sqrt{2})|ab|
\label{eq:oA}
\end{eqnarray}
Then
$$
p(a,b) = p(a,e) \oplus \Pa_R^{-1}(b' \rightarrow a) \oplus  \Pa_R^{-1}(b \rightarrow c)
$$
is a path from $a$ to $b$ of length
\begin{eqnarray*}
\nonumber |p(a,b)| & \le & |p(a,e)| + |b'a|\sqrt{2} + |bc|\sqrt{2} ~~~~~~~~~~~~~~~~~\mbox{(by~(\ref{eq:PR}))} \\
                   & \le & (27 + 20\sqrt{2})|ab| + 3\sqrt{2}|ab| + 2|ab| ~~~~~~~~\mbox{(by~(\ref{eq:o4b}),~(\ref{eq:o1}))} \\
                   & = & (29 + 23\sqrt{2})|ab|
\end{eqnarray*}

\subsection{$Y_k$ is a Spanner, for $k \ge 7$}
\begin{lemma}
Let $\theta$ be a real number with $0 < \theta < \pi/3$, % and $k \ge 1$,
and let
\[ t = \frac{1 + \sqrt{2 - 2 \cos\theta}}{2 \cos \theta -1} .
\]
Let $a$, $b$, and $c$ be three distinct points in the plane such that
$|ac| \leq |ab|$, let $\alpha = \angle{bac}$, and assume that
$0 \leq \alpha \leq \theta$. Then
\begin{equation}  \label{eqtoprove}
   |bc| \leq |ab| - |ac|/t .
\end{equation}
\label{lem:60degrees}
\end{lemma}
\begin{proof}
Refer to Figure~\ref{fig:60degrees}. By the Law of Cosines, we have
\[ |bc|^2 = |ac|^2 + |ab|^2 - 2 |ac| \cdot |ab| \cos \alpha .
\]
%
%%%%%%%%%%%%%%%%%%%%%%%%%%%%%%%%%Figure Begin
\begin{figure}[htbp]
\centering
\includegraphics[width=0.26\linewidth]{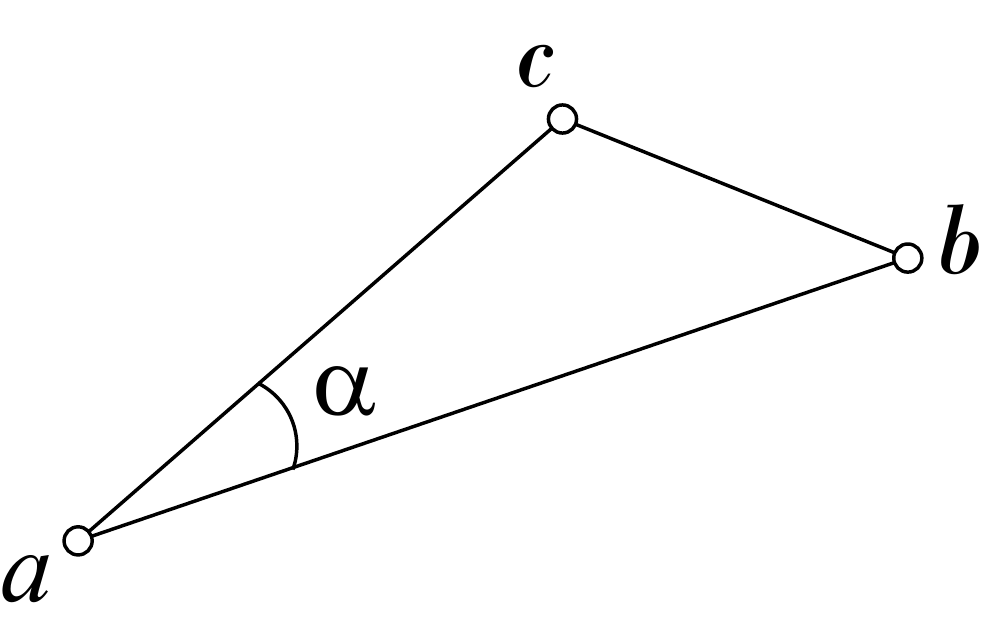}
\caption{Lemma 1: If $\alpha < 60$ and $|ac| \leq |ab|$, then $|bc| \leq |ab| - |ac|/t$.}
\label{fig:60degrees}
\end{figure}
%%%%%%%%%%%%%%%%%%%%%%%%%%%%%%%%%Figure End
%
Since $t>1$ and $|ac| \leq |ab|$, the right-hand side in
(\ref{eqtoprove}) is positive, so (\ref{eqtoprove}) is equivalent
to
\[ |bc|^2 \leq \left( |ab| - |ac|/t \right)^2 .
\]
Thus, we have to show that
\[ |ac|^2 + |ab|^2 - 2 |ac| \cdot |ab| \cos \alpha
      \leq \left( |ab| - |ac|/t \right)^2 ,
\]
which simplifies to
\begin{equation}     \label{eqtoshow}
    \left( 1 - 1/t^2 \right) |ac|
       \leq 2 (\cos \alpha - 1/t) |ab| .
\end{equation}
Since $|ac| \leq |ab|$ and $\cos \theta \leq \cos \alpha$,
(\ref{eqtoshow}) holds if
\[ 1 - 1/t^2 \leq 2 (\cos \theta - 1/t) ,
\]
which can be rewritten as
\begin{equation}     \label{eqtoshow3}
   ( 2 \cos \theta - 1 ) t^2 - 2t + 1 \geq 0 .
\end{equation}
By our choice of $t$, equality holds in (\ref{eqtoshow3}).
\end{proof}

\noindent
An immediate consequence of Lemma~\ref{lem:60degrees} is the
following result.
%(established in~\cite{RS91} for
%$t=\frac{1}{1-2\sin\frac{\theta}{2}}$.)
%{\bf REMARK:} This is not true: they prove this for the
%$\Theta$-graph.
\begin{theorem}
For any $\theta$ with $0 < \theta < \pi/3$, the Yao-graph with cones of angle
$\theta$, is a $t$-spanner for
\[ t = \frac{1 + \sqrt{2 - 2 \cos\theta}}{2 \cos \theta -1} .
\]
\label{thm:Y6}
\end{theorem}
\begin{proof}
The proof of this claim is by induction on the distances defined by
the $n \choose 2$ pairs of nodes. Since $\theta < \pi/3$, any closest
pair is connected by an edge in the Yao-graph; this proves the
basis of the induction. The induction step follows from
Lemma~\ref{lem:60degrees}.
\end{proof}

What happens to the value of $t$ from Lemma~\ref{lem:60degrees},
if $\theta$ gets close to $\pi/3$:
Let $\varepsilon = \cos \theta - 1/2$, so that $\varepsilon$ is close to zero.
Then
\begin{eqnarray*}
 t & = & \frac{1}{2 \varepsilon} +
         \sqrt{ \frac{1 - 2 \varepsilon}{4 \varepsilon^2} } \\
   & = & \frac{1}{2 \varepsilon} +
         \frac{\sqrt{1 - 2 \varepsilon}}{2 \varepsilon}  \\
   & \sim & \frac{1}{2 \varepsilon} + \frac{1 - \varepsilon}{2 \varepsilon} \\
   & = & - \frac{1}{2} + \frac{1}{\varepsilon} \\
   & = & - \frac{1}{2} + \frac{1}{\cos \theta - 1/2} .
\end{eqnarray*}

\end{document}